\documentclass[journal]{IEEEtran}
\ifCLASSINFOpdf
\else
\fi

\usepackage{graphics} 
\usepackage{graphicx}

\usepackage{times} 
\usepackage{amsmath} 
\usepackage{amssymb}  

\usepackage{tikz} 
\usepackage{pgfplots} 

\usepackage{algpseudocode}%
\usepackage[ruled,vlined,linesnumbered]{algorithm2e}
\usepackage{epstopdf}
\usepackage{cases}
\usepackage{bm}

\usepackage{color}

\usepackage{lipsum}
\usepackage{verbatim}

\usepackage{amsthm}
\usepackage{siunitx}

\usepackage{mathtools}

\newtheorem{theorem}{\bf Theorem} \newtheorem{definition}[theorem]{\bf Definition} 
\newtheorem{lemma}[theorem]{\bf Lemma} \newtheorem{remark}[theorem]{\bf Remark}
\newtheorem{corollary}[theorem]{\bf Corollary}  
\newtheorem{assumption}[theorem]{\bf Assumption}  

\newtheorem{problem}{\bf Problem}

\newcommand{\I}{\mathbb{I}}
\newcommand{\R}{\mathbb{R}}

\newcommand{\N}{\mathbb{N}}

\colorlet{istorange}{orange}
\colorlet{istgreen}{green!50!black}
\colorlet{istblue}{blue} 
\colorlet{istred}{red!90!black}

\title{Data-driven analysis and controller design for discrete-time systems under aperiodic sampling}
\author{Stefan Wildhagen, Julian Berberich, Michael Hertneck and Frank Allg\"ower
	\thanks{This work was funded by the Deutsche Forschungsgemeinschaft (DFG, German Research Foundation) under 285825138 and under Germany’s Excellence Strategy - EXC 2075 - 390740016. J. Berberich thanks the International Max Planck Research School for Intelligent Systems (IMPRS-IS) for supporting him. The authors are with the University of Stuttgart, Institute for Systems Theory and Automatic Control, 70569 Stuttgart, Germany (email:
		{\{wildhagen}{,berberich}{,hertneck}{,allgower\}}{@ist.uni-stuttgart.de}).}}

\begin{document}

\maketitle

\begin{abstract}
	This article is concerned with data-driven analysis of discrete-time systems under aperiodic sampling, and in particular with a data-driven estimation of the maximum sampling interval (MSI). The MSI is relevant for analysis of and controller design for cyber-physical, embedded and networked systems, since it gives a limit on the time span between sampling instants such that stability is guaranteed. We propose tools to compute the MSI for a given controller and to design a controller with a preferably large MSI, both directly from a finite-length, noise-corrupted state-input trajectory of the system. We follow two distinct approaches for stability analysis, one taking a robust control perspective and the other a switched systems perspective on the aperiodically sampled system. In a numerical example and a subsequent discussion, we demonstrate the efficacy of our developed tools and compare the two approaches.
\end{abstract}


\section{Introduction} \label{sec:introduction}
\IEEEPARstart{T}{he} widespread and ever-increasing prevalence of cyber-physical systems (CPSs) and of embedded and networked control systems (NCSs) has sparked a large interest in the study of sampled-data control systems in the recent decades \cite{chen1995optimal,hetel2017recent}. Especially \textit{aperiodically} sampled systems emerged as a tool to abstract and analyze many real-world scenarios where the sampling time points are not guaranteed to be, or are even designed \textit{not} to be, equidistant: NCSs with packet dropouts, delays and aperiodic scheduling strategies, or event- and self-triggered approaches, to name a few examples. A fundamental concept for systems under aperiodic sampling is the maximum sampling interval (MSI), i.e., the largest time span between two sampling instants such that stability is preserved. Knowledge or at least a good approximation of the MSI is often required for analysis and design of aperiodic sampling strategies. A multitude of different methods for analyzing aperiodically sampled systems and estimating the MSI have been discussed in the literature, such as the time-delay approach \cite{fridman2004robust,fridman2010refined,seuret2018wirtinger}, the impulsive/hybrid systems approach \cite{nagshtabrizi2008sampled,carnevale2007lyapunov}, the switched systems approach \cite{hetel2011discrete,xiong2007packet_loss,yu2004dropout} or the robust input/output approach \cite{mirkin2007some,fujioka2009IQC}.

All of the methods mentioned above require an accurate model to be applicable, although obtaining such a model via first principles can be a challenging task. Measured trajectories of a system, by contrast, can typically be obtained easily. This fact has been leveraged to estimate a model based on given data in the field of system identification~\cite{ljung1987system}, and more recently, to perform system analysis and controller design directly via measured data \cite{hou2013model}.
An interesting stream of research, which is closely related to the results in this paper, has been sparked by the finding that the entire behavior of a linear time-invariant (LTI) system can be described by a single trajectory~\cite{willems2005note}. A number of different works considered topics such as verifying dissipativity properties \cite{maupong2017lyapunov,koch2021provably}, model predictive control \cite{coulson2019deepc,berberich2021guarantees} or state-feedback design \cite{persis2020formulas,berberich2020design,waarde2020from}, all based on measured data, as well as combining data with prior knowledge for controller design \cite{berberich2020combining}.

The main contribution of this work is to investigate data-driven analysis of discrete-time systems under aperiodic sampling. In contrast to existing model-based sampled-data control approaches which usually handle continuous-time systems, we deal with discrete-time systems, since, in any practical scenario, data can only be measured at discrete time instants. In addition, the discrete-time perspective is justified by the fact that CPSs, embedded control systems and NCSs involve digital devices, which dictate a clock rate for the entire control system. In particular, we develop techniques a) to estimate the MSI for a given controller and b) to design controllers with a preferably high MSI bound, both directly from a noise-corrupted, finite-length state-input trajectory of the unknown system. We present two different approaches to achieve these goals: The first, referred to as robust input/output approach, is based on splitting the aperiodically sampled system into an LTI system and a delay operator, and a subsequent stability analysis using input/output properties of the delay operator and robust control tools; the second, referred to as switched systems approach, comprehends the aperiodically sampled system with its sampling period-dependent dynamics as a switched system, and leverages the well-understood analysis for this system class. As we will see later, the robust input/output approach has a lower computational complexity and comes with little conservatism even if the available data are noisy, whereas the switched systems approach can outperform the latter and provide very tight estimations of the MSI especially when rich data are available. In both approaches, we leverage a recently proposed data-driven system parametrization \cite{waarde2020from,berberich2020combining} to verify the stability conditions robustly for all systems consistent with the data.

To the best of our knowledge, this is the first work to offer a data-driven analysis of discrete-time aperiodically sampled systems. Our recent work \cite{berberich2021aper_samp} solves a similar problem in the continuous-time domain, combining data-driven control methods with the time-delay approach to sampled-data systems \cite{fridman2010refined}. Note that sampled-data systems form a subclass of time-delay systems with bounded delay. For this reason, the results on data-driven control of time-delay systems in \cite{rueda2021delay} could be used to treat the problem considered in this paper as well. As another alternative, one could identify the unknown system from the measured data, and in a second step use the identified model in existing model-based discrete-time stability conditions for aperiodically sampled systems \cite{hetel2011discrete,xiong2007packet_loss,yu2004dropout,seuret2012novel} in a two-step procedure. However, we note that our proposed approaches are more direct, allowing to estimate the MSI and to design a controller directly from data in a single step. Furthermore, obtaining tight estimation bounds from noisy data of finite length is a challenging problem when estimating a model, e.g., via least-squares estimation \cite{matni2019self}, which might be problematic since we ultimately aim for guarantees for the true underlying system. Nonetheless, methods based on set membership estimation \cite{milanese1991optimal,belforte1990parameter} are indeed guaranteed to contain the true system in the error bounds. We compare both mentioned alternatives, \cite{rueda2021delay} and the two-step procedure based on set membership estimation, to our proposed approaches in a numerical example in Section \ref{sec:num}.

The remainder of this article is structured as follows. In Section \ref{sec:setup}, we introduce the system under aperiodic sampling, the available data and the considered problems. The robust input/output approach and the switched systems approach can be found in Sections \ref{sec:IO} and \ref{sec:switched}, respectively. A numerical analysis of the proposed approaches is performed in Section \ref{sec:num}. We summarize this article and give potential directions for future work in Section \ref{sec:summary}.

We denote by $\N$ the set of natural numbers, $\N_0\coloneqq \N\cup \{0\}$ and $\N_{[a,b]}\coloneqq\N_0\cap[a,b]$, $\N_{\ge a} \coloneqq \N_0\cap[a,\infty)$, $a,b\in\N_0$. We denote by $I$ the identity matrix and by $0$ the zero matrix of appropriate dimension, and by $I_n$ the identity matrix of dimension $n$. Let $A\in\R^{n\times n}$ be a real matrix. We write $A\succ0$ $(A\succeq 0)$ if $A$ is symmetric and positive (semi-)definite, and we denote negative (semi-)definiteness similarly. Let $\sigma_\text{max}(A)$ denote the maximum singular value of $A$. We write $\lVert v\rVert_2$ for the 2-norm of a vector $v\in\R^n$, and $\lVert A\rVert_2$ for the induced 2-norm of $A$. The Hermitian transpose of a complex matrix $B\in\mathbb{C}^{n\times m}$ is denoted by $B^*$. The Kronecker product of two matrices $C\in\R^{n\times m}$ and $D\in\R^{p\times r}$ is denoted by $C\otimes D$. We write $\star$ if an element in a matrix can be inferred from symmetry. We denote by $\ell_2$ the space of square integrable signals and by $\ell_{2e}$ the extended $\ell_2$ space. For some $T\in\N_0$, we denote by $\cdot_T:\ell_{2e}\to\ell_{2e}$ the truncation operator, which assigns to a signal $y\in\ell_{2e}$ the signal $y_T$ which satisfies $y_T(t)=y(t)$ for all $t\in\N_{[0,T]}$ and $y_T(t)=0$ for all $t\in\N_{\ge T+1}$. We write $\lVert y\rVert_{\ell_2}$ for the $\ell_2$ norm of a signal $y\in\ell_{2}$ and $\lVert \Delta\rVert_{\ell_2}\coloneqq\inf\{\gamma\;\vert\; \lVert\Delta(y)_T\rVert_{\ell_2}\le\gamma\lVert y_T\rVert_{\ell_2}, \; y\in\ell_{2e}, \; T\in\N_0\}$ for the $\ell_2$ gain of an operator $\Delta: \ell_{2e}\to\ell_{2e}$.

\section{Setup and Problem Statement} \label{sec:setup}

\subsection{System under aperiodic sampling}

Let us consider a discrete-time LTI system
\begin{equation} \label{eq:system}
x(t+1) = A_\text{tr} x(t) + B_\text{tr} u(t), \; x(0) = x_0 \in \R^n
\end{equation}
with state $x(t)\in\R^n$, input $u(t)\in\R^m$ and time instants $t\in\N_0$. A typical application for such a discrete-time setup are continuous-time processes that are sensed, controlled and actuated by digital devices, which is commonly the case in CPSs, embedded control systems and NCSs. We assume throughout this article that the true system matrices $A_\text{tr}$, $B_\text{tr}$ are \textit{unknown} and that only state-input measurements are available.

To close the loop, \eqref{eq:system} is sampled and controlled at aperiodic \textit{sampling instants} $t_k\in\N_0$, where
\begin{equation}
t_0 = 0, \quad t_{k+1}-t_k \ge 1,
\end{equation}
such that the \textit{sampling interval} $h_k\coloneqq t_{k+1}-t_k$ is time-varying. A sampled version of the plant state is $\{x(t_k)\}_{k=0}^{\infty}$, from which the controller computes a sequence of control values $\{u(t_k)\}_{k=0}^{\infty}$ via a linear state-feedback law $u(t_k)=K x(t_k)$, $K\in\R^{m\times n}$. We assume that the control input applied to the plant is held constant in between sampling instants, i.e., $u(t)=u(t_k), \; t\in \N_{[t_k,t_{k+1}-1]}$. The aperiodically sampled system in closed loop can then be written as
\begin{equation} \label{eq:system_aper_sampled}
\begin{aligned}
x(t\hspace{-1pt}+\hspace{-1pt}1) &= A_\text{tr} x(t) \hspace{-1pt}+\hspace{-1pt} B_\text{tr} K x(t_k), \: \forall t\in\N_{[t_k,t_{k+1}\hspace{-0.5pt}-\hspace{-0.5pt}1]}, \: \forall k\in\N_{0} \\
t_{k+1} &= t_k + h_k, \; \forall k\in\N_{0} \\
t_0 &= 0, \quad x(0) = x_0.
\end{aligned}
\end{equation}

Aperiodic sampling is relevant for a multitude of practical scenarios. For instance, NCSs, where the communication between plant and controller takes place over a shared communication network, exhibit several phenomena that give rise to aperiodic sampling: Packet dropouts, which occur especially in wireless and/or congested networks, result in varying sampling periods \cite{xiong2007packet_loss}. Further, the reception of control updates can be out of order due to time-varying transmission delays, in which case it might be useful to discard the outdated ones \cite{hespanha2007survey,zhang2016survey}. Another common scenario are aperiodic or contention-based network access protocols. In addition, it is often favorable not to sample and transmit the system's state periodically to counteract overloading of the network, as is done for instance in (periodic) event-triggered and self-triggered approaches \cite{heemels2012introduction,heemels2012periodic}. In all these cases, it is possible to write the resulting control system with aperiodic sampling in the form \eqref{eq:system_aper_sampled}.

In the above-named scenarios, although the exact sampling instants are unknown in advance, a bound on the sampling interval is often known or can be estimated. This is the case, e.g., if the number of consecutive packet losses and the maximum delay are upper bounded \cite{xiong2007packet_loss,zhang2016survey}. In event- and self-triggered schemes, the inter-event times can be proven to be upper bounded under certain conditions \cite{postoyan2019inter,gleizer2021bisimulation}. Further, in practical implementations of such schemes, it is often sensible to constrain the length of the sampling interval to guarantee a certain amount of attention dedicated to the process \cite{gleizer2021bisimulation}. Hence, in this article, we are concerned with stability analysis and controller design for \eqref{eq:system_aper_sampled} in case the sampling interval is arbitrary but bounded, i.e., $h_k$ takes arbitrary values in the set $\N_{[1,\overline{h}]}$ for a given $\overline{h}\in\N$. In addition, we are interested in estimating a preferably tight lower bound on the largest $\overline{h}$ such that \eqref{eq:system_aper_sampled} is asymptotically stable, i.e., on the MSI.
\begin{remark}
	In most classical works on aperiodically sampled systems, the sampling interval takes values in $h_k\in(0,\infty)$ (cf. \cite{hetel2017recent,fridman2004robust,fridman2010refined,nagshtabrizi2008sampled,carnevale2007lyapunov,hetel2011discrete,mirkin2007some,fujioka2009IQC}), as continuous-time systems are considered therein.	In contrast, we consider $h_k\in\N$ in this article since we work in discrete time. As a result, the considered setup can be seen as the discrete-time equivalent to the classical approaches formulated in continuous time.
\end{remark}

\subsection{Available data} \label{sec:setup_data}

The main challenge in this article is that analysis and controller design are performed without knowledge of the true matrices $A_\text{tr}$ and $B_\text{tr}$. Instead, we suppose that state-input data $\{x(t)\}_{t=0}^{N}$, $\{u(t)\}_{t=0}^{N-1}$, $N\in\N$, of the perturbed system
\begin{equation} \label{eq:system_data_perturbed}
x(t+1) = A_\text{tr} x(t) + B_\text{tr} u(t) + B_d d(t)
\end{equation}
are available, where $d(t)\in\R^{n_d}$ is an unknown disturbance and $B_d$ is a known matrix. $B_d$ can be used to incorporate knowledge on the way the disturbance enters the system, e.g., if it is known that it only affects a subset of the states. If no such prior knowledge is available, then one may set $B_d=I$.

The particular disturbance sequence $\{\hat{d}(t)\}_{t=0}^{N-1}$ that affected the measured data is unknown and written in matrix form as $\hat{D}\coloneqq \begin{bmatrix} \hat{d}(0) & \cdots & \hat{d}(N-1)\end{bmatrix}$. As in \cite{berberich2020combining}, we assume that a multiplier description for the disturbance is available.
\begin{assumption} \label{ass:disturbance_bound}
	The disturbance satisfies $\hat{D}\in\mathcal{D}$, where
	\begin{align*}
	\mathcal{D}\coloneqq\Big\{D\in\mathbb{R}^{n_d\times N}\Bigm|
	\begin{bmatrix}D^\top\\I\end{bmatrix}^\top
	P_d
	\begin{bmatrix}D^\top\\I\end{bmatrix}\succeq0, \; \forall P_d\in\boldsymbol{P}_d \Big\}, 
	\end{align*}
	and where $\boldsymbol{P}_d$ is a convex cone of symmetric matrices admitting an LMI representation. Moreover, there exists $P_d\in\boldsymbol{P}_d$ such that $\begin{bmatrix}I & 0\end{bmatrix}P_d\begin{bmatrix}I & 0\end{bmatrix}^\top\prec 0$.
\end{assumption}
Assumption \ref{ass:disturbance_bound} encompasses various special cases, e.g., quadratic bounds on the entire sequence
\begin{equation*}
\begin{bmatrix}\hat{D}^\top\\I\end{bmatrix}^\top
\begin{bmatrix}Q_d & S_d \\ S_d^\top & R_d \end{bmatrix}
\begin{bmatrix}\hat{D}^\top\\I\end{bmatrix}\succeq0, \; Q_d\prec 0,
\end{equation*}
as used similarly, e.g., in \cite{koch2021provably,persis2020formulas,berberich2020design,waarde2020from,berberich2021aper_samp}, via
\begin{equation} \label{eq:quad_multiplier}
\boldsymbol{P}_d = \Big\{\tau\begin{bmatrix}Q_d & S_d \\ S_d^\top & R_d \end{bmatrix}\Bigm|
\tau > 0 \Big\}.
\end{equation}
Alternatively, componentwise 2-norm bounds $\lVert\hat{d}(t)\rVert_2\le\overline{d}$ can be considered via diagonal multipliers
\begin{equation} \label{eq:diag_multiplier}
	{\begingroup
	\setlength\arraycolsep{1pt}
	\boldsymbol{P}_d \hspace{-1pt}=\hspace{-1pt} \Big\{\hspace{-3pt}\begin{bmatrix}-\text{diag}\{p_i\}_{i=1}^N & 0 \\ 0 & \sum_{i=1}^{N} p_i\overline{d}^2 I_{n_d} \end{bmatrix}\hspace{-2pt}\Bigm|\hspace{-1pt}
	p_i \ge 0, \:  i\in\N_{[1,N]} \hspace{-1pt}\Big\}.
	\endgroup}
\end{equation}
Such bounds have also been considered for the purpose of robust data-driven control in \cite{bisoffi2021trade}. The latter is especially useful for describing the practically relevant case of point-wise disturbance bounds, since it provides a tight description thereof and, other than quadratic full-block bounds, guarantees that the set of matrices compatible with the data does not grow when more data is added. Let $c_d(N)$ denote the number of decision variables involved in the respective multiplier description. Then, implementation of diagonal multipliers involves $c_d(N)=N$ decision variables compared to $c_d(N)=1$ for quadratic multipliers. A thorough discussion of Assumption \ref{ass:disturbance_bound} and possible choices of $\boldsymbol{P}_d$ can be found in \cite[Subsection II.C]{berberich2020combining}.

Furthermore, we pose the following assumption on $B_d$.
\begin{assumption} \label{ass:Bd}
	The matrix $B_d$ has full column rank.
\end{assumption}
This is essentially without loss of generality: Should it not be the case, one can define another pair $d',B_{d'}$ with the same influence on \eqref{eq:system_data_perturbed}, but with $B_{d'}$ satisfying Assumption \ref{ass:Bd} \cite{berberich2020combining}.
\begin{remark}
	As an alternative to the disturbance description introduced in \eqref{eq:system_data_perturbed} and Assumptions \ref{ass:disturbance_bound} and \ref{ass:Bd}, one could also consider a formulation involving measurement noise $w$ (cf. \cite[Subsection V.A]{persis2020formulas}). If a multiplier description of $A_\text{tr}W-W^+$ (with $W,W^+$ defined similarly as in \eqref{eq:data_matrices}) was available, the results presented in this paper would require only minor modifications to accommodate such a setup.
\end{remark}

Finally, we note that the measurements are taken at each of the time instants. It is not restrictive to assume that such data are available, since in sampled-data systems, the aperiodic sampling comes into play only for closed-loop operation. By contrast, the required state and input trajectories can be collected independently of each other in an open-loop experiment. To illustrate this, consider an NCS where communication with the controller takes place via a network. In such a setup, it is nonetheless possible to probe the system with an open-loop input trajectory and to record the system response at the sensor without any need for using the network.

\subsection{Problem statement}

Having introduced aperiodically sampled systems and the available data, we may now formalize the problems considered in this paper.
\begin{problem}[Aperiodic Sampling: Analysis] \label{prob:arbitrary_sampling_an}
	Given state-input measurements $\{x(t)\}_{t=0}^{N}$, $\{u(t)\}_{t=0}^{N-1}$ of \eqref{eq:system_data_perturbed}, a disturbance description $\mathcal{D}$, a matrix $B_d$, an upper bound $\overline{h}$ for the sampling interval, and furthermore, a controller $K$, determine if the origin of the closed-loop aperiodically sampled system \eqref{eq:system_aper_sampled} is asymptotically stable for an arbitrarily time-varying sampling interval $h_k\in\N_{[1,\overline{h}]}$.
\end{problem}
\begin{problem}[Aperiodic Sampling: Controller Design] \label{prob:arbitrary_sampling_co}
	Given the same setup as in Problem \ref{prob:arbitrary_sampling_an}, design a controller $K$ such that the origin of the closed-loop aperiodically sampled system \eqref{eq:system_aper_sampled} is asymptotically stable for an arbitrarily time-varying sampling interval $h_k\in\N_{[1,\overline{h}]}$.
\end{problem}
\begin{problem}[Maximum Sampling Interval] \label{prob:MSI}
	Determine the largest $\overline{h}$, denoted by $\overline{h}_\text{MSI}$, such that Problem \ref{prob:arbitrary_sampling_an}, respectively Problem \ref{prob:arbitrary_sampling_co}, admits a solution. In other words, determine a possibly tight lower bound of the MSI.
\end{problem}
We will focus on Problems \ref{prob:arbitrary_sampling_an} and \ref{prob:arbitrary_sampling_co} in the technical sections of this article. This is because once we have obtained a solution to Problems \ref{prob:arbitrary_sampling_an} and \ref{prob:arbitrary_sampling_co}, a solution to Problem \ref{prob:MSI} can be found via a linear search over $\overline{h}$ \cite[Chapter 6]{knuth1997art}, or, often more efficiently, via an exponential search \cite{bentley1976almost}. We note that since $\overline{h}$ takes integer values, these algorithms are guaranteed to be successful \cite[Chapter 6]{knuth1997art}, i.e., to return the tightest MSI bound.

We propose two approaches to address Problems \ref{prob:arbitrary_sampling_an} and \ref{prob:arbitrary_sampling_co}. In the first approach (Section \ref{sec:IO}), we write the aperiodically sampled system \eqref{eq:system_aper_sampled} as an interconnection of a classical LTI system and a delay operator. In the second approach (Section \ref{sec:switched}), we comprehend \eqref{eq:system_aper_sampled} as a discrete-time switched system.

In practice, one often deals with continuous-time processes. In such a case, the choice of \emph{discretization period} $H\in(0,\infty)$, with which the data are recorded from the underlying physical process, is of great importance. The reason is that while it is often possible to stabilize the discrete-time system \eqref{eq:system_data_perturbed} with a large discretization period, the achieved control performance might be very poor. Especially in a data-driven setup, i.e., when the dynamics of the process to-be-controlled are not known, it might be hard to adequately select $H$ such that all relevant dynamics are reflected in the discrete-time model and a certain performance objective is met.

However, in practice often expert knowledge is available, e.g., an estimate of the fastest relevant dynamics of the system, from which a suitable $H$ can be chosen. Alternatively, one could explicitly specify a performance objective and incorporate it into the stability conditions. Then, one could test if the desired performance objective can be satisfied, and if not, adjust the discretization period $H$ accordingly.

Although we focus on the stabilization problem (i.e., Problems \ref{prob:arbitrary_sampling_an} and \ref{prob:arbitrary_sampling_co}) in the main body of this article, we will demonstrate in Remark \ref{rem:perf} that a performance objective can be directly incorporated in the robust input/output approach. For the switched systems approach, incorporating a performance objective is not straightforward due to the unknown inter-sampling behavior, and is left to future work.

\section{Robust input/output approach} \label{sec:IO}

The existing (model-based) robust input/output approach to aperiodically sampled systems, pioneered in \cite{fridman2004robust,mirkin2007some,kao2004stability_delay} and further refined in \cite{fujioka2009IQC}, is tailored to continuous-time systems. To accommodate our discrete-time setup, which we consider because we ultimately aim for an incorporation of data, we provide input/output methods for analyzing discrete-time systems in this section.

In particular, we take the following steps. In Subsection \ref{sec:IO_setup}, we introduce the main idea of the robust input/output approach, which is first to write the aperiodically sampled system as a time-delay system which exhibits a sawtooth shape of the delay sequence. Second, following \cite{kao2012discrete_delayed_IQC}, we express this time-delay system as an interconnection of an LTI system and a delay operator. In Subsection \ref{sec:IO_delay_operator}, we provide a novel bound on the $\ell_2$ gain of this delay operator. With this, model-based stability conditions are formulated using results from robust control theory in Subsection \ref{sec:IO_stab_model}. Finally, in Subsection \ref{sec:IO_stab_data}, we present data-driven stability criteria by combining the model-based conditions with the data-driven system parametrization in \cite{berberich2020combining}.

We use the well-established concept of integral quadratic constraints (IQCs) (see \cite{kao2012discrete_delayed_IQC,megretski1997system,veenman2016IQC,hu2017discreteIQC,scherer2021dissipativity}) in order to describe input/output properties of the delay operator and to establish stability conditions. A definition of an IQC in the input/output framework, which we will use throughout this paper, is given below.
\begin{definition} \label{def:IQC}
	A bounded, causal operator $\Delta:\ell_{2e}^p\to\ell_{2e}^q$, $y\mapsto e$ satisfies the hard static\footnote{\emph{Hard} refers to the fact that the IQC holds for all $T\ge 0$, whereas \emph{soft} IQCs hold only for $T\to\infty$. \emph{Static} refers to the fact that the signals $(y,\Delta y)$ are multiplied with $\Pi$ directly, whereas filtered versions thereof are used in \emph{dynamic} IQCs. A static IQC is equivalent to dissipativity with a general quadratic supply rate (cf. \cite{scherer2000linear,kottenstette2014relationships,scherer2021dissipativity}). Nonetheless, we speak of IQCs in this paper to be consistent with the literature on time-delay and aperiodically sampled systems \cite{hetel2017recent,fujioka2009IQC,kao2012discrete_delayed_IQC}.} IQC defined by a multiplier $\Pi\in\R^{(p+q) \times (p+q)}$ if the following inequality holds for all $y\in\ell_2^p$
	\begin{equation} \label{eq:def_IQC}
	\sum_{t=0}^{T}
	\begin{bmatrix} y(t) \\ (\Delta y)(t) \end{bmatrix}^\top
	\Pi \begin{bmatrix} y(t) \\ (\Delta y)(t) \end{bmatrix} \ge 0 , \quad\forall T\in\N_0.
	\end{equation}
\end{definition}
As a short-hand notation, we write $\Delta\in\text{IQC}(\Pi)$ if $\Delta$ satisfies the hard static IQC \eqref{eq:def_IQC} in the sense of Definition \ref{def:IQC}.

\subsection{Main idea} \label{sec:IO_setup}

Note that, as an intermediate step, we can rewrite the aperiodically sampled system \eqref{eq:system_aper_sampled} as a time-delay system
\begin{equation} \label{eq:system_delayed}
\begin{aligned}
x(t+1) &= (A_\text{tr} + B_\text{tr}K) x(t) + B_\text{tr}K (x(t_k)-x(t)) \\
&= (A_\text{tr} + B_\text{tr}K) x(t) + B_\text{tr}K (x(t-\tau(t))-x(t))
\end{aligned}
\end{equation}
for $t\in\N_{[t_k,t_{k+1}-1]}$, $k\in\N_{0}$, where
\begin{equation}
\tau(t)\coloneqq t-t_k, \; t\in\N_{[t_k,t_{k+1}-1]}, \; k\in\N_{0},
\end{equation}
represents the amount of time by which the feedback information is delayed at time $t$. Note that $\tau(0)=0$ since $t_0=0$. The \textit{delay sequence} $\{\tau(t)\}$ has a ``sawtooth shape'' as illustrated in Figure \ref{fig:delay_sawtooth}, i.e., it is reset to zero at sampling instants and is subsequently increased by one in each time step where no sampling occurs. The maximum possible delay is $\overline{h}-1$.

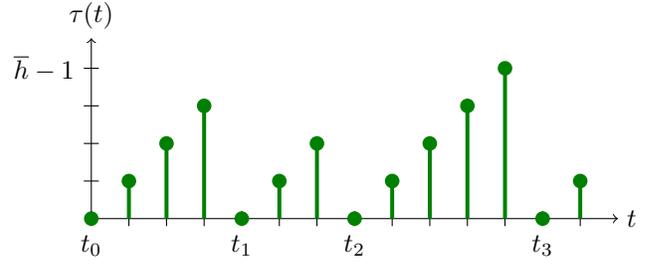
\begin{figure}
	\begin{tikzpicture}
	\draw[->] (0,0) -- (7,0) node[right] {$t$};
	\draw[->] (0,0) -- (0,2.4) node[above] {$\tau(t)$};
	\draw (-0.1,2) node[left] {$\overline{h}-1$} -- (0.1,2);
	\draw (0,-0.1) node[below] {$t_0$} -- (0,0.1);
	\draw (2,-0.1) node[below] {$t_1$} -- (2,0.1);
	\draw (3.5,-0.1) node[below] {$t_2$} -- (3.5,0.1);
	\draw (6,-0.1) node[below] {$t_3$} -- (6,0.1);
	
	\foreach \x in {0.5,1,1.5,2.5,3,4,4.5,5,5.5,6.5}
	\draw (\x,-0.1) -- (\x,0.1);
	
	\foreach \y in {0.5,1,1.5}
	\draw (-0.1,\y) -- (0.1,\y);
	
	\draw[color=istgreen,line width = 1.5pt] plot[ycomb,mark=*] coordinates {(0,0) (0.5,0.5) (1,1) (1.5,1.5) (2,0)  (2.5,0.5)(3,1) (3.5,0) (4,0.5) (4.5,1) (5,1.5) (5.5,2) (6,0) (6.5,0.5) };
	\end{tikzpicture}
	\caption{Sawtooth shape of $\tau(t)$.}
	\label{fig:delay_sawtooth}
\end{figure}

Following \cite[Section 2]{kao2012discrete_delayed_IQC}, we subsequently rewrite the time-delay system \eqref{eq:system_delayed} as an interconnection between an LTI system and a delay operator. First, we note that
\begin{equation*}
e(t)\coloneqq x(t-\tau(t))-x(t)
\end{equation*}
can be viewed as the ``error'' induced by aperiodic sampling. It can be represented by the telescopic sum
\begin{align*}
e(t) & = x(t-\tau(t))-x(t) =  x(t-\tau(t)) \nonumber \\
&-x(t-\tau(t)+1)+x(t-\tau(t)+1)-\ldots-x(t) \nonumber \\
&= \sum_{i=t-\tau(t)}^{t-1} x(i)-x(i+1).
\end{align*}
If we introduce an artificial output
\begin{align*}
y(t) &\coloneqq x(t)-x(t+1) \\
&= (I-A_\text{tr}-B_\text{tr}K)x(t)-B_\text{tr}Ke(t)
\end{align*}
and define the \textit{delay operator} $\Delta:\ell_{2e}^n\to\ell_{2e}^n$, $e=\Delta y$, as
\begin{equation} \label{eq:def_delta}
e(t) = (\Delta y)(t) \coloneqq \sum_{i= t-\tau(t)}^{t-1} y(i), \; t\in\N_{[t_k,t_{k+1}-1]}, \; k\in\N_{0},
\end{equation}
we may write the aperiodically sampled system \eqref{eq:system_aper_sampled} as an interconnection of an LTI system and the delay operator
\begin{subequations} \label{eq:fb_interconnection_ss}
	\begin{align}
	\begin{bmatrix}
	x(t+1)\\ y(t)
	\end{bmatrix}&=
	\begin{bmatrix}
	A_\text{tr}+B_\text{tr}K&B_\text{tr}K\\
	I-A_\text{tr}-B_\text{tr}K&-B_\text{tr}K
	\end{bmatrix}
	\begin{bmatrix}
	x(t)\\ e(t)\end{bmatrix}, \label{eq:fb_interc_syst} \\[3pt]
	e(t)&=(\Delta y)(t). \label{eq:fb_interc_delay}
	\end{align}
\end{subequations}

The main idea is now, as in \cite{kao2012discrete_delayed_IQC}, to analyze stability of the feedback interconnection \eqref{eq:fb_interconnection_ss} using robust control theory. In particular, the delay operator will be embedded into a class of uncertainties acting on the LTI system \eqref{eq:fb_interc_syst} by bounding its $\ell_2$ gain and thereby describing its input/output behavior as a hard static IQC.

\subsection{$\ell_2$ gain of the delay operator} \label{sec:IO_delay_operator}

Subsequently, we provide an estimate of the $\ell_2$ gain of the delay operator $\Delta$ and express this condition in terms of a hard static IQC. In contrast to \cite{kao2012discrete_delayed_IQC}, which analyzed time-delay systems with an arbitrary bounded delay, in case of aperiodically sampled systems the delay $\{\tau(t)\}$ exhibits a sawtooth shape as illustrated in Figure \ref{fig:delay_sawtooth}. By exploiting this information, we will be able to derive a smaller $\ell_2$-gain estimate than if we would dissolve $\{\tau(t)\}$ into the class of bounded delays $\tau(t)\in\N_{[0,\overline{h}-1]}$ and use the results in \cite{kao2012discrete_delayed_IQC}. This is the main technical contribution in terms of the model-based analysis of aperiodically sampled discrete-time systems as presented in Subsections \ref{sec:IO_setup}-\ref{sec:IO_stab_model}. It is important to note in this respect that a tighter uncertainty description of $\Delta$ will lead to less conservative stability conditions for the interconnection \eqref{eq:fb_interconnection_ss}.

\begin{lemma} \label{lem:L2_gain}
	The $\ell_2$ gain of the delay operator $\Delta$ is upper bounded by $\sqrt{\frac{\overline{h}}{2}(\overline{h}-1)}$, i.e., for all $y\in\ell_2^n$, we have
	\begin{equation*}
	\sum_{t=0}^{T} (\Delta y)(t)^\top (\Delta y)(t) \le \frac{\overline{h}}{2}(\overline{h}-1) \sum_{t=0}^{T} y(t)^\top y(t), \quad \forall T\in\N_0.
	\end{equation*}
\end{lemma}
\noindent The proof can be found in Appendix \ref{app:proof_lem_L2}.

\begin{remark}
	Taking the bounded-delay perspective as in \cite{kao2012discrete_delayed_IQC}, one may find the bound $\overline{h}-1$ on the $\ell_2$ gain of $\Delta$. Lemma~\ref{lem:L2_gain} provides a tighter bound, since it holds that $\sqrt{\frac{\overline{h}}{2}(\overline{h}-1)}\le \overline{h}-1$ for all $\overline{h}\in\N$. In particular, equality holds for $\overline{h}\in\{1,2\}$, whereas the strict inequality comes into play for $\overline{h}\in\N_{\ge 3}$. As $\overline{h}\to\infty$, the improvement approaches a factor of $\frac{1}{\sqrt{2}}$.
\end{remark}

\begin{remark} \label{rem:L2gain_cont}
	In the continuous-time case, the factor of improvement when taking into account the sawtooth shape compared to taking the bounded delay perspective, was proven to be $\frac{2}{\pi}$ for all $\overline{h}$ \cite{mirkin2007some}. In contrast, the factor of improvement is dependent on $\overline{h}$ in the discrete-time case.
\end{remark}

The $\ell_2$ gain proven in Lemma \ref{lem:L2_gain} directly implies that $\Delta$ satisfies a hard static IQC.

\begin{corollary} \label{cor:L2_gain_IQC}
	For any $\mathcal{X}=\mathcal{X}^\top \succ 0$, it holds that $\Delta\in\mathrm{IQC}(\Pi_{\ell_2})$, where
	\begin{equation*}
	\Pi_{\ell_2} \coloneqq	\begin{bmatrix}	\frac{\overline{h}}{2}(\overline{h}-1)\mathcal{X} & 0 \\ 0 & -\mathcal{X} \end{bmatrix}.
	\end{equation*}
\end{corollary}

\begin{proof}
	A positive definite $\mathcal{X}$ allows a factorization $\mathcal{X}=\mathcal{X}^{\frac{1}{2}}\mathcal{X}^\frac{1}{2}$ where $\mathcal{X}^{\frac{1}{2}}\succ 0$. Further, it holds that $\Delta (\mathcal{X}^{\frac{1}{2}}y) = \mathcal{X}^{\frac{1}{2}} \Delta (y)$, and that $\mathcal{X}^{\frac{1}{2}}y\neq 0$ if $y\neq 0$ since $\mathcal{X}^\frac{1}{2}\succ 0$. With this and Lemma \ref{lem:L2_gain}, the hard static IQC follows immediately.
\end{proof}

\subsection{Model-based stability criteria} \label{sec:IO_stab_model}

In this subsection we assume, in contrast to the remainder of this article, that the true system matrices $A_\text{tr}$ and $B_\text{tr}$ are known. To formulate stability conditions for \eqref{eq:fb_interconnection_ss}, we exploit the hard static IQC for the delay operator in Corollary \ref{cor:L2_gain_IQC}, which allows us to directly use existing results on linear systems in feedback with uncertainties satisfying hard static IQCs.
\begin{theorem} \label{thm:stab_IO_ss}
	Suppose there exist matrices $\mathcal{Q}=\mathcal{Q}^\top\succ 0\in\R^{n\times n}$ and $\mathcal{X} = \mathcal{X}^\top \succ 0\in\R^{n\times n}$ such that \eqref{eq:stab_cond_IO_ss} is satisfied. Then, the origin of \eqref{eq:fb_interconnection_ss} is asymptotically stable.
	\begin{figure*}[h]
		\vspace{2pt}
		\begin{align}\label{eq:stab_cond_IO_ss}
		\left[
		\begin{array}{cc}
		A_\text{tr}+B_\text{tr}K&B_\text{tr}K\\I&0\\\hline
		I-A_\text{tr}-B_\text{tr}K&-B_\text{tr}K\\0&I
		\end{array}
		\right]^\top
		\left[
		\begin{array}{c|c}
		\begin{matrix}\mathcal{Q}&0\\0&\mathcal{-Q}\end{matrix}&
		\begin{matrix}0&0\\0&0\end{matrix}\\\hline
		\begin{matrix}0&0\\0&0\end{matrix}&\begin{matrix}	\frac{\overline{h}}{2}(\overline{h}-1)\mathcal{X} & 0 \\ 0 & -\mathcal{X} \end{matrix}
		\end{array}
		\right]
		\left[
		\begin{array}{cc}
		A_\text{tr}+B_\text{tr}K&B_\text{tr}K\\I&0\\\hline
		I-A_\text{tr}-B_\text{tr}K&-B_\text{tr}K\\0&I
		\end{array}
		\right]
		\prec0
		\end{align}
		\noindent\makebox[\linewidth]{\rule{\textwidth}{0.4pt}}
	\end{figure*}
\end{theorem}
\begin{proof}
	The statement of the theorem will be a direct consequence of \cite[Corollary 11]{scherer2021dissipativity}, which is a modification of the classical circle criterion \cite{zames1968stability}. The proof of \cite[Corollary 11]{scherer2021dissipativity} relies on recognizing that \eqref{eq:stab_cond_IO_ss} implies that \eqref{eq:fb_interc_syst} is dissipative w.r.t. $\Pi_{\ell_2}$ and that \eqref{eq:fb_interconnection_ss} forms a so-called neutral interconnection, from which stability follows. Thus, we now verify the assumptions of \cite[Corollary 11]{scherer2021dissipativity}. First, we note that the delay operator satisfies the hard static IQC $\Delta\in\text{IQC}(\Pi_{\ell_2})$ as established in Corollary \ref{cor:L2_gain_IQC}. Second, the delay operator satisfies \cite[Condition (8)]{scherer2021dissipativity}, as $\left[\substack{y\\0}\right]^\top \Pi_{\ell_2} \left[\substack{y\\0}\right]\ge 0$ holds due to $\mathcal{X}\succ 0$.  Lastly, we note that the stability conditions for continuous-time systems in \cite[Corollary 11]{scherer2021dissipativity} can be used for discrete-time systems as well under the modification discussed in \cite[Page 12, Ramifications]{scherer2021dissipativity}.
\end{proof}

\subsection{Data-driven stability criteria for analysis and controller design} \label{sec:IO_stab_data}

Finally, in this subsection, we tackle Problems \ref{prob:arbitrary_sampling_an} and \ref{prob:arbitrary_sampling_co}. Recall that state-input measurements $\{x(t)\}_{t=0}^N$, $\{u(t)\}_{t=0}^{N-1}$ of the disturbed system \eqref{eq:system_data_perturbed} are available. Note that the noise-corrupted data could be explained by a multitude of different matrices $A,B$, especially since the particular disturbance sequence that corrupted the experiment is unknown as well. Let us arrange the measured data as follows
\begin{align}
X^+&=\begin{bmatrix}x(1)&x(2)&\dots&x(N)\end{bmatrix}, \nonumber\\
X&=\begin{bmatrix}x(0)&x(1)&\dots&x(N-1)\end{bmatrix}, \label{eq:data_matrices}\\
U&=\begin{bmatrix}u(0)&u(1)&\dots&u(N-1)\end{bmatrix}. \nonumber
\end{align}
With this, we may define the set of all matrices compatible with the measured data and the disturbance bound as
\begin{align*}
\Sigma_{AB}=\{\begin{bmatrix} A & B \end{bmatrix}\mid X^+=AX+BU+B_dD,D\in\mathcal{D}\}.
\end{align*}
In order to guarantee stability for \eqref{eq:fb_interconnection_ss}, we need to verify the model-based stability condition for each $\begin{bmatrix} A & B \end{bmatrix}\in\Sigma_{AB}$. To this end, we use a data-driven parametrization of the compatible matrices as developed in \cite{berberich2020combining}, where it was shown that it is possible to represent the matrices contained in $\Sigma_{AB}$ by a quadratic constraint. We define
\begin{align*}
\boldsymbol{P}_{AB} \coloneqq\left[\begin{array}{cc}
-X&0\\-U&0\\\hline X^+&B_d
\end{array}\right]
\boldsymbol{P}_d
\left[\begin{array}{cc}
	-X&0\\-U&0\\\hline X^+&B_d
\end{array}\right]^\top.
\end{align*}
Then, it follows directly from \cite[Lemma 2]{berberich2020combining} and the discussion thereafter that the set $\Sigma_{AB}$ can be expressed in terms of a quadratic constraint as follows
\begin{equation} \label{eq:param_IO}
\Sigma_{AB}=\Big\{\begin{bmatrix} A & B \end{bmatrix}\hspace{-2pt}\Bigm|\hspace{-2pt}
\begin{bmatrix}A^\top\\B^\top\\I\end{bmatrix}^\top \hspace{-4pt}
P_{AB} \hspace{-1pt}
\begin{bmatrix}A^\top\\B^\top\\I\end{bmatrix}\succeq0, \: \forall P_{AB}\in\boldsymbol{P}_{AB}\Big\}.
\end{equation}

The idea is now to rewrite the interconnection of the LTI system and the delay operator as a linear fractional transformation (LFT) with \textit{two} uncertainty channels
\begin{subequations}\label{eq:system_IO_uncertain}
	\begin{align}
	\left[\begin{array}{c}
	x(t+1)\\\hline y(t)\\z(t)
	\end{array}\right]&=
	\left[\begin{array}{c|cc}
	0&0&I\\\hline
	I&0&-I\\
	\begin{bmatrix}I\\K\end{bmatrix}&\begin{bmatrix}0\\K\end{bmatrix}&0
	\end{array}\right]
	\left[\begin{array}{c}
	x(t)\\\hline e(t)\\w(t)\end{array}\right],\\
	e(t)&=(\Delta y)(t),\\
	w(t)&=(\begin{bmatrix} A & B \end{bmatrix}z)(t),
	\end{align}
\end{subequations}
where $\Delta\in\text{IQC}(\Pi_{\ell_2})$ and $\begin{bmatrix} A & B \end{bmatrix}\in\Sigma_{AB}$. The first channel $e\mapsto y$ contains the delay operator as in the model-based case (cf. \eqref{eq:fb_interconnection_ss}), and the second channel $z\mapsto w$ represents the uncertainty in the system matrices due to the data-driven setup and the disturbance. Using this system description, we are finally able to provide a data-driven stability condition for \eqref{eq:system_IO_uncertain}, using the parametrization \eqref{eq:param_IO} of the compatible system matrices and the S-procedure.
\begin{theorem} \label{thm:stab_IO_data}
	Suppose Assumptions \ref{ass:disturbance_bound} and \ref{ass:Bd} are satisfied and $\overline{h}\ge 2$. Furthermore suppose, given a controller $K$, there exist matrices $\mathcal{S}=\mathcal{S}^\top\succ 0\in\R^{n\times n}$, $\mathcal{X}_\text{inv}=\mathcal{X}_\text{inv}^\top\succ 0\in\R^{n\times n}$ and $P_{AB}\in\boldsymbol{P}_{AB}$ such that \eqref{eq:stab_cond_IO_data} is satisfied. Then, the origin of \eqref{eq:system_IO_uncertain} is asymptotically stable for any $\begin{bmatrix}
	A & B \end{bmatrix} \in\Sigma_{AB}$.
	\begin{figure*}
		\vspace{2pt}
		\begin{align}\label{eq:stab_cond_IO_data}
		\left[\begin{array}{ccc}
		0&I&\begin{bmatrix}I&K^\top\end{bmatrix}\\I&0&0\\\hline
		0&0&\begin{bmatrix}0&K^\top\end{bmatrix}\\0&I&0\\\hline
		0&0&I\\I&-I&0
		\end{array}
		\right]^\top
		\left[
		\begin{array}{c|c|c}
		\begin{matrix}\mathcal{S}&0\\0&\mathcal{-S}\end{matrix}&
		\begin{matrix}0&0\\0&0\end{matrix}&\begin{matrix}0&0\\0&0\end{matrix}\\\hline
		\begin{matrix}0&0\\0&0\end{matrix}&\begin{matrix} \mathcal{X}_\text{inv} & 0 \\ 0 & -{\tfrac{2}{\overline{h}(\overline{h}-1)}} \mathcal{X}_\text{inv} \end{matrix}&\begin{matrix}0&0\\0&0\end{matrix}\\\hline
		\begin{matrix}0&0\\0&0\end{matrix}&\begin{matrix}0&0\\0&0\end{matrix}&
		P_{AB}
		\end{array}
		\right]
		\left[\begin{array}{ccc}
		0&I&\begin{bmatrix}I&K^\top\end{bmatrix}\\I&0&0\\\hline
		0&0&\begin{bmatrix}0&K^\top\end{bmatrix}\\0&I&0\\\hline
		0&0&I\\I&-I&0
		\end{array}
		\right]\prec0
		\end{align}
		\noindent\makebox[\linewidth]{\rule{\textwidth}{0.4pt}}
	\end{figure*}
\end{theorem}
\begin{proof}
	We use the full-block S-procedure \cite[Lemma A.1]{scherer2000robust} to conclude that \eqref{eq:stab_cond_IO_data} implies that
	\begin{align} \label{eq:thm_io_data_sproc}
	{\footnotesize
		\begingroup
		\setlength\arraycolsep{1pt}
		\left[
		\begin{array}{cc}
		(A+BK)^\top&(BK)^\top\\I&0\\\hline
		(I-A-BK)^\top&-(BK)^\top\\0&I
		\end{array}
		\right]^\top
		\left[
		\begin{array}{c|c}
		\begin{matrix}\mathcal{S}&0\\0&\mathcal{-S}\end{matrix}&
		\begin{matrix}0&0\\0&0\end{matrix}\\\hline
		\begin{matrix}0&0\\0&0\end{matrix}&\begin{matrix} \mathcal{X}_\text{inv} & 0 \\ 0 & -{\tfrac{2}{\overline{h}(\overline{h}-1)}} \mathcal{X}_\text{inv} \end{matrix}
		\end{array}\right]
		\star
		\endgroup
		\prec0}
	\end{align}
	holds for all $\begin{bmatrix} A & B \end{bmatrix}$ which satisfy \eqref{eq:param_IO}, i.e., for all $\begin{bmatrix} A & B \end{bmatrix}\in\Sigma_{AB}$. It is easy to verify that \eqref{eq:thm_io_data_sproc} is equivalent to
	\begin{align*}
		{\footnotesize
			\begingroup
			\setlength\arraycolsep{1pt}
			\left[
			\begin{array}{cc}
				(A+BK)^\top&(BK)^\top\\I&0\\\hline
				(I-A-BK)^\top&-(BK)^\top\\0&I
			\end{array}
			\right]^\top
			\left[
			\begin{array}{c|c}
				\begin{matrix}0&-I\\I&0\end{matrix}&
				\begin{matrix}0&0\\0&0\end{matrix}\\\hline
				\begin{matrix}0&0\\0&0\end{matrix}&\begin{matrix} 0 & -I \\ I & 0 \end{matrix}
			\end{array}\right]^\top\cdot\endgroup} \\
			{\footnotesize
				\begingroup
				\setlength\arraycolsep{1pt}
				\left[
			\begin{array}{c|c}
				\begin{matrix}\mathcal{S}&0\\0&\mathcal{-S}\end{matrix}&
				\begin{matrix}0&0\\0&0\end{matrix}\\\hline
				\begin{matrix}0&0\\0&0\end{matrix}&\begin{matrix} {\tfrac{2}{\overline{h}(\overline{h}-1)}} \mathcal{X}_\text{inv} & 0 \\ 0 & - \mathcal{X}_\text{inv} \end{matrix}
			\end{array}\right]
			\star
			\endgroup
			\succ0,}
	\end{align*}
	which, in turn, is equivalent to
	\begin{align} \label{eq:thm_io_data_sproc2}
		{\footnotesize
			\begingroup
			\setlength\arraycolsep{1pt}
			\left[
			\begin{array}{cc}
				I&0\\-(A+BK)^\top&-(BK)^\top\\\hline
				0&I\\(A+BK-I)^\top&(BK)^\top
			\end{array}
			\right]^\top
			\left[
			\begin{array}{c|c}
				\begin{matrix}\mathcal{S}&0\\0&\mathcal{-S}\end{matrix}&
				\begin{matrix}0&0\\0&0\end{matrix}\\\hline
				\begin{matrix}0&0\\0&0\end{matrix}&\begin{matrix} {\tfrac{2}{\overline{h}(\overline{h}-1)}} \mathcal{X}_\text{inv} & 0 \\ 0 & - \mathcal{X}_\text{inv} \end{matrix}
			\end{array}\right]
			\star
			\endgroup
			\succ0.}
	\end{align}
	We are now ready to apply the dualization lemma \cite[Lemma 4.9]{scherer2000linear} to \eqref{eq:thm_io_data_sproc2}. The required inertia assumptions clearly hold since the inner matrix in \eqref{eq:thm_io_data_sproc2} possesses exactly $2n$ positive and negative eigenvalues. Using the dualization lemma, \eqref{eq:thm_io_data_sproc2} directly implies that
	\begin{align}
		{\footnotesize
			\begingroup
			\setlength\arraycolsep{1pt}
			\left[
			\begin{array}{cc}
			A+BK&BK\\	I&0\\\hline
			I-A-BK&-BK\\0&I
			\end{array}
			\right]^\top
			\hspace{-3pt}
			\left[
			\begin{array}{c|c}
				\begin{matrix}\mathcal{S}^{-1}&0\\0&\mathcal{-S}^{-1}\end{matrix}&
				\begin{matrix}0&0\\0&0\end{matrix}\\\hline
				\begin{matrix}0&0\\0&0\end{matrix}&\begin{matrix} \frac{\overline{h}}{2}(\overline{h}-1)\mathcal{X}_\text{inv}^{-1} & 0 \\ 0 & -\mathcal{X}_\text{inv}^{-1} \end{matrix}
			\end{array}\right]
			\hspace{1pt}
			\star
			\hspace{-1pt}
			\endgroup
			\prec0}
	\end{align}
	holds for any $\begin{bmatrix} A & B \end{bmatrix}\in\Sigma_{AB}$. With Theorem \ref{thm:stab_IO_ss} and taking $\mathcal{Q}\coloneqq\mathcal{S}^{-1}$ and $\mathcal{X}\coloneqq\mathcal{X}_\text{inv}^{-1}$ in \eqref{eq:stab_cond_IO_ss}, we finally conclude asymptotic stability for any $\begin{bmatrix} A & B \end{bmatrix}\in\Sigma_{AB}$.
\end{proof}

\begin{remark}
	Theorem \ref{thm:stab_IO_data} requires $\overline{h}\ge 2$ to ensure that the inverse of $\Pi_{\ell_2}$ exists. Note that if $\overline{h}=1$, the aperiodically sampled system \eqref{eq:system_aper_sampled} would in fact degenerate to a periodically sampled one, in which case the results in \cite{berberich2020combining} or Theorem \ref{thm:stab_switched_data_an} (see Subsection \ref{sec:switched_stability}) can be used to analyze stability.
\end{remark}

\begin{remark} \label{rem:perf}
	A performance objective can be incorporated into the robust input/output approach, e.g., by adding the performance channel $d\to p$ to the LFT \eqref{eq:system_IO_uncertain}, where
	\begin{equation*}
		p(t) = C_p x(t) + D_p u(t).
	\end{equation*}
	An $\mathcal{H}_2$-performance bound of $\gamma>0$ for the channel $d\to p$ could then be enforced (cf. \cite[Theorem 17]{waarde2020from}, \cite[Theorem 1]{berberich2020combining}) if \eqref{eq:stab_cond_IO_data} holds with $-\mathcal{S}$ in the second diagonal element of the middle matrix replaced by $B_d B_d^\top-\mathcal{S}$, and if there additionally exists $\Gamma$ such that $\text{trace}(\Gamma)<\gamma^2$ and
	\begin{equation*}
		\begin{bmatrix}
			\Gamma & (C_p+D_pK)\mathcal{S} \\
			\star & \mathcal{S}
		\end{bmatrix} \succ 0.
	\end{equation*}
	An $\mathcal{H}_\infty$ performance objective could be incorporated analogously (cf. \cite[Theorem 20]{waarde2020from}, \cite[Theorem 2]{berberich2020combining}).
\end{remark}

Condition \eqref{eq:stab_cond_IO_data} in Theorem \ref{thm:stab_IO_data} allows for a simultaneous search for $\mathcal{S}$, $\mathcal{X}_\text{inv}$ and $P_{AB}$ via a semi-definite program (SDP), since it is an LMI in these variables. However, \eqref{eq:stab_cond_IO_data} is not amenable to controller design yet. For a particular choice of $\mathcal{X}_\text{inv}$, namely $\mathcal{X}_\text{inv}=\mathcal{S}$, a simultaneous search for all involved variables and a stabilizing controller $K$ is indeed possible via an SDP using the standard transformation $\mathcal{F}=K\mathcal{S}$, as we present in the following result.
\begin{corollary} \label{cor:stab_IO_data}
	Suppose Assumptions \ref{ass:disturbance_bound} and \ref{ass:Bd} are satisfied and $\overline{h}\ge 2$. Furthermore, suppose there exist matrices $\mathcal{S}=\mathcal{S}^\top\succ 0\in\R^{n\times n}$, $\mathcal{F}\in\R^{m\times n}$ and $P_{AB}=\left[\substack{Q_{AB} \; S_{AB} \\ S_{AB}^\top \; R_{AB}}\right]\in\boldsymbol{P}_{AB}$ such that 
	\begingroup 
	\setlength\arraycolsep{2.5pt}
	\begin{equation} {\footnotesize \label{eq:stab_cond_IO_data_1}
		\begin{bmatrix} R_{AB} - \mathcal{S} & \star & \star & \star \\
	 -R_{AB} & \frac{\overline{h}(\overline{h}-1)-2}{\overline{h}(\overline{h}-1)}\mathcal{S} + R_{AB} & \star & \star \\
	 S_{AB} & \begin{bmatrix} \mathcal{S} \\ \mathcal{F} \end{bmatrix} - S_{AB} & \begin{bmatrix} \mathcal{S} & \mathcal{F}^\top \\ \mathcal{F} & 0 \end{bmatrix} + Q_{AB} & \star \\
		0 & 0 & \begin{bmatrix} 0 & \mathcal{F}^\top \end{bmatrix} & -\frac{1}{2} \mathcal{S}
		\end{bmatrix}\prec 0}
	\end{equation}
	\endgroup
	is satisfied. Then, the controller $K\coloneqq \mathcal{F}\mathcal{S}^{-1}$ renders the origin of \eqref{eq:system_IO_uncertain} asymptotically stable for any $\begin{bmatrix} A & B \end{bmatrix} \in\Sigma_{AB}$.
\end{corollary}
\begin{proof}
	We apply the Schur complement to \eqref{eq:stab_cond_IO_data_1} with respect to its fourth diagonal component and substitute $\mathcal{F} = K\mathcal{S}$ to obtain \eqref{eq:stab_cond_IO_data} with $\mathcal{X}_\text{inv} = \mathcal{S}$. This can be easily seen by calculating $\tilde{\Pi}$ for this particular choice and subsequently expanding  \eqref{eq:stab_cond_IO_data}. Stability then follows from Theorem \ref{thm:stab_IO_data}.
\end{proof}

Naturally, since the true system matrices $A_\text{tr},B_\text{tr}$ are contained in $\Sigma_{AB}$, the aperiodically sampled system \eqref{eq:system_aper_sampled} is stable if \eqref{eq:system_IO_uncertain} is robustly stable for all $\begin{bmatrix} A & B \end{bmatrix} \in\Sigma_{AB}$. For completeness, this statement is formalized in the following result.
\begin{corollary} \label{cor:stab_IO_data_true}
	Suppose the conditions of Theorem \ref{thm:stab_IO_data} or of Corollary \ref{cor:stab_IO_data} are fulfilled. Then the controller $K$ renders the origin of the aperiodically sampled system \eqref{eq:system_aper_sampled} asymptotically stable for all $h_k\in\N_{[1,\overline{h}]}$.
\end{corollary}

Theorem \ref{thm:stab_IO_data} may be used to address Problem \ref{prob:arbitrary_sampling_an}, since it allows to check an MSI bound for a given controller. Corollary \ref{cor:stab_IO_data}, where $\mathcal{X}_\text{inv}$ is coupled to $\mathcal{S}$, enables a co-search for a controller and thereby solves Problem 2.

\section{Switched systems approach} \label{sec:switched}

In this section, we interpret the aperiodically sampled system, whose dynamics depend on the arbitrarily time-varying sampling interval, as a switched system. We elaborate on the main idea of this approach before stating our contributions.

\subsection{Main idea} \label{sec:switched_basic_idea}

Let us define for an $h\in\N$
\begin{equation*}
A^h_\text{tr} \coloneqq (A_\text{tr})^h \text{ and } B^h_\text{tr}\coloneqq\sum_{i=0}^{h-1} (A_\text{tr})^iB_\text{tr}.\end{equation*}
Then, it is immediate that the state of the aperiodically sampled system \eqref{eq:system_aper_sampled} at sampling instants is governed by the discrete-time switched system
\begin{equation} \label{eq:system_switched}
x(t_{k+1}) = (A^{h_k}_\text{tr} + B^{h_k}_\text{tr}K) x(t_k),
\end{equation}
where the switching sequence is given by $h_k\in\N_{[1,\overline{h}]}$. Due to linearity of the aperiodically sampled system \eqref{eq:system_aper_sampled}, the state at inter-sampling instants is bounded in the sense of \cite[Definition 2]{nesic1999formulas}, such that asymptotic stability of \eqref{eq:system_switched} implies asymptotic stability of \eqref{eq:system_aper_sampled} and vice versa \cite[Theorem 2]{nesic1999formulas}. Note that since $A_\text{tr}$, $B_\text{tr}$ are unknown, the system matrices appearing in the switched system $A^h_\text{tr}$, $B^h_\text{tr}$, $h\in\N_{[1,\overline{h}]}$, are naturally unknown as well. 

Model-based stability analysis of switched systems is a well-understood topic \cite{liberzon1999basic,daafouz2002switched,liberzon2003switched_book} and the framework has been successfully applied to aperiodically sampled systems\footnote{To be precise, \cite{xiong2007packet_loss,yu2004dropout} considered an NCS with bounded packet loss, which admits an equivalent mathematical description.} \cite{hetel2011discrete,xiong2007packet_loss,yu2004dropout}. In particular, it was shown in \cite{xiong2007packet_loss,daafouz2002switched} that stability holds if there exist matrices $\mathcal{S}^h=\mathcal{S}^{h\top}\succ 0\in\R^{n\times n}$ and $\mathcal{G}^h\in\R^{n\times n}$, $h\in\N_{[1,\overline{h}]}$, which satisfy
\begin{equation} \label{eq:stab_cond_switched_ss_2}
\begin{bmatrix}
\mathcal{G}^h + \mathcal{G}^{h\top} - \mathcal{S}^h & \star \\
(A_\text{tr}^h+B_\text{tr}^h K)\mathcal{G}^h & \mathcal{S}^j
\end{bmatrix} \succ 0, \; (h,j)\in \N_{[1,\overline{h}]}^2.
\end{equation}

However, it is not straightforward to translate \eqref{eq:stab_cond_switched_ss_2} to the data-driven setup. This is mainly owed to the fact that the matrices $A_\text{tr}^h$, $B_\text{tr}^h$, which correspond to the different possible sampling intervals, appear in the stability conditions. From the given trajectory $\{x(t)\}_{t=0}^{N}$, $\{u(t)\}_{t=0}^{N-1}$, it is not directly possible to obtain data-driven parametrizations of $B^h$, since this would require that the input in the recorded data is held constant over $h$ consecutive time instants. A possibility to obtain the sought parametrizations would be to conduct multiple experiments with all the sampling periods of interest. However, this could potentially be very laborious and costly.
	
In Subsection \ref{sec:switched_parametrization}, we show that by reexpressing \eqref{eq:system_switched} with a \textit{lifted} input matrix, it is in fact possible to obtain parametrizations of the required matrices for all sampling periods $h\in\N_{[1,\overline{h}]}$. Our solution uses the given data $\{x(t)\}_{t=0}^{N}$, $\{u(t)\}_{t=0}^{N-1}$ only and requires a slightly modified assumption on the disturbance. In Subsection \ref{sec:switched_stability}, we provide data-driven conditions for stability using the derived parametrizations.

\subsection{Data-driven parametrization of the lifted systems} \label{sec:switched_parametrization}

Note that from the definition of $B^h_\text{tr}$, it follows that we can rewrite its product with $K$ as
\begin{equation*}
B^h_\text{tr} K = \begin{bmatrix} A^{h-1}_\text{tr} B_\text{tr} & \cdots & B_\text{tr} \end{bmatrix} \begin{bmatrix} K \\ \vdots \\ K \end{bmatrix}.
\end{equation*}
Hence, defining the lifted input matrix and the lifted controller
\begin{equation} \label{eq:lifted_system_matrices}
\underline{B}^h_\text{tr}\coloneqq \begin{bmatrix} A^{h-1}_\text{tr} B_\text{tr} & \cdots & B_\text{tr} \end{bmatrix},
\; \begin{rcases} \underline{K}^h \coloneqq \begin{bmatrix} K \\ \vdots \\ K \end{bmatrix} \hspace{5pt} \end{rcases} h \text{ times},
\end{equation}
the switched system \eqref{eq:system_switched} can be rewritten as the lifted switched system
\begin{equation}\label{eq:system_switched_lifted}
x(t_{k+1}) = \left(A_\text{tr}^{h_k} + \underline{B}^{h_k}_\text{tr} \underline{K}^{h_k} \right) x(t_k). 
\end{equation}

Recall that we have access to one measured state-input trajectory $\{x(t)\}_{t=0}^{N}$, $\{u(t)\}_{t=0}^{N-1}$ of \eqref{eq:system_data_perturbed}. The recorded data were sampled at each time instant and were affected by the unknown disturbance sequence $\{\hat{d}(t)\}_{t=0}^{N-1}$. We will show in this section that using this data, it is possible to obtain data-driven parametrizations of the lifted matrices $\begin{bmatrix} A^h & \underline{B}^h \end{bmatrix}$ appearing in \eqref{eq:system_switched_lifted}. Let us define, for an arbitrary $h\in\N_{[1,\overline{h}]}$, the matrices containing the measured data (see also \cite{xin2021self}) as
\begin{align*}
&X^h_+\coloneqq \begin{bmatrix} x(h) & x(h+1) & \cdots & x(N)\end{bmatrix}, \\
&X^h \coloneqq \begin{bmatrix} x(0) & x(1) & \cdots & x(N-h) \end{bmatrix}, \\				
&U^h \coloneqq \begin{bmatrix} u(0) & u(1) & \cdots & u(N-h) \\
\vdots & \vdots &  & \vdots \\
u(h-1) & u(h) & \cdots & u(N-1)	\end{bmatrix}.
\end{align*}

With the assumption that the measured data stem from the perturbed linear system \eqref{eq:system_data_perturbed}, it is immediate that the matrices $X^h_+$, $X^h$ and $U^h$ contain state-input measurements of the lifted perturbed system
\begin{equation} \label{eq:system_lifted_perturbed}
\begin{aligned}
x(t+h) &= A^{h}_\text{tr} x(t) + \underline{B}_\text{tr}^{h} \begin{bmatrix} u(t) \\ \vdots \\ u(t+h-1) \end{bmatrix} \\
&+\begin{bmatrix} A_\text{tr}^{h-1} B_d & \cdots & B_d \end{bmatrix}
\begin{bmatrix} d(t) \\ \vdots \\ d(t+h-1) \end{bmatrix},
\end{aligned}
\end{equation}
$h\in\N_{[1,\overline{h}]}$. Note that the system matrices of the lifted perturbed system \eqref{eq:system_lifted_perturbed} are the same as in the lifted switched system \eqref{eq:system_switched_lifted}. We define the lifted disturbance for $h=1$ as $D^1=\begin{bmatrix} d^1(0) & \cdots & d^1(N-1) \end{bmatrix}\coloneqq \begin{bmatrix} d(0) & \cdots & d(N-1) \end{bmatrix}$ and for
$h\in\N_{[2,\overline{h}]}$ as
\begin{equation} \label{eq:lifted_disturbance}
\begin{aligned}
&D^h = \begin{bmatrix} d^h(0) & \cdots & d^h(N-h) \end{bmatrix} \\
&\coloneqq\begin{bmatrix} {(A_\text{tr}^{h-1}B_d)}^\top \\ \vdots \\ B_d^\top \end{bmatrix}^\top\hspace{-1.5pt}
\begin{bmatrix} d(0) & \cdots & d(N-h) \\
\vdots &  & \vdots \\
d(h-1) & \cdots & d(N-1) \end{bmatrix}\hspace{-2pt}.
\end{aligned}
\end{equation}
We can now define the set of all lifted matrices consistent with the measured data and the lifted disturbance \eqref{eq:lifted_disturbance} as
\begin{equation*}
\begin{aligned}
\Sigma^h_{AB}\coloneqq\{&\begin{bmatrix} A^h & \underline{B}^h \end{bmatrix}\in \R^{n\times (n+hm)} \mid \\
& X_+^h=A^h X^h + \underline{B}^h U^h +B_d^h D^h, D^h\in\mathcal{D}^h\},
\end{aligned}
\end{equation*}
for all $h\in\N_{[1,\overline{h}]}$, where $B_d^1\coloneqq B_d$, $\mathcal{D}^1 \coloneqq \mathcal{D}$, $B_d^h \coloneqq I$, and
\begin{equation*}
\mathcal{D}^h\coloneqq \{D^h \text{ according to \eqref{eq:lifted_disturbance}} \hspace{-0.6pt}\bigm|\hspace{-0.6pt} \begin{bmatrix} d(0) & \cdots & d(N-1) \end{bmatrix}\in\mathcal{D} \}
\end{equation*}
for all $h\in\N_{[2,\overline{h}]}$. Clearly, the particular lifted disturbance sequence that affected the measurements also lies in $\mathcal{D}^h$.

In order to obtain a data-driven parametrization of $\Sigma_{AB}^h$ as in \cite{berberich2020combining}, we require a description of the set $\mathcal{D}^h$ via a quadratic constraint. Denoting $n_d^1\coloneqq n_d$, $n_d^h=n$ for $h\in\N_{[2,\overline{h}]}$, and $N_h\coloneqq N-h+1$, $\mathcal{D}^h$ would need to satisfy the representation
\begin{align} \label{eq:lifted_dist_QMI}
\Big\{D^h\in\mathbb{R}^{n_d^h\times N_h}\Bigm|
\begin{bmatrix}D^{h\top}\\I\end{bmatrix}^\top
P_d^h
\begin{bmatrix}D^{h\top}\\I\end{bmatrix}\succeq0, \; \forall P_d^h\in\boldsymbol{P}_d^h \Big\}
\end{align}
where $\boldsymbol{P}_d^h$ is a known convex cone of symmetric matrices admitting an LMI representation, for which there exists $P_d^h\in\boldsymbol{P}_d^h$ that satisfies $\begin{bmatrix}I & 0\end{bmatrix}P_d^h\begin{bmatrix}I & 0\end{bmatrix}^\top\prec 0$. For $h\geq2$, it is generally difficult to derive a tight characterization as in \eqref{eq:lifted_dist_QMI}, even if $A_\text{tr}$ was known.

Therefore, in the following, we will derive an \textit{overapproximation} of the set $\mathcal{D}^h$ by a quadratic constraint for all $h\in\N_{[1,\overline{h}]}$. We assume the following special case for the disturbance bound $\mathcal{D}$.
\begin{assumption} \label{ass:comp_wise_disturbance_bound}
	The components of the unknown disturbance sequence $\hat{D}=\begin{bmatrix} \hat{d}(0) & \cdots & \hat{d}(N-1) \end{bmatrix}$ are norm-bounded in the 2-norm with a known upper bound $\overline{d}\ge 0$, i.e.,
	\begin{equation*}
	\lVert d(t)\rVert_2\le\overline{d}, \; \forall t\in\N_{[0,N-1]}.
	\end{equation*}
\end{assumption}
For such disturbance sequences, we derive the preferably tight overapproximation
\begingroup
\setlength\arraycolsep{2.5pt}
\begin{align} \label{eq:switched_disturbance_overapprox}
\tilde{\mathcal{D}}^h\hspace{-1pt}\coloneqq\hspace{-1pt}\Big\{\hspace{-1pt}\tilde{D}^h\in\mathbb{R}^{n_d^h\times N_h}\hspace{-2pt}\Bigm| \hspace{-3pt}
\begin{bmatrix}\tilde{D}^{h\top}\\I\end{bmatrix}^\top \hspace{-5pt}
\tilde{P}_d^h \hspace{-2pt}
\begin{bmatrix}\tilde{D}^{h\top}\\I\end{bmatrix}\hspace{-1pt}\succeq\hspace{-1pt}0, \: \forall \tilde{P}_d^h\in\tilde{\boldsymbol{P}}_d^h\Big\},
\end{align}
\endgroup
which corresponds to determining a class of multipliers $\tilde{\boldsymbol{P}}_d^h$ with $\tilde{P}_d^h\in\tilde{\boldsymbol{P}}_d^h$ such that $\begin{bmatrix}I & 0\end{bmatrix}\tilde{P}_d^h\begin{bmatrix}I & 0\end{bmatrix}^\top\prec 0$, for which $\mathcal{D}^h \subseteq \tilde{\mathcal{D}}^h$ is satisfied for all $h\in\N_{[1,\overline{h}]}$. Then, if we define
\begin{equation*}
\begin{aligned}
\tilde{\Sigma}^h_{AB}\coloneqq\{&\begin{bmatrix} A^h & \underline{B}^h \end{bmatrix}\in \R^{n\times (n+hm)} \mid \\
& X_+^h=A^h X^h + \underline{B}^h U^h + B_d^h \tilde{D}^h, \tilde{D}^h\in\tilde{\mathcal{D}}^h\},
\end{aligned}
\end{equation*}
it is clear that $\Sigma^h_{AB} \subseteq \tilde{\Sigma}^h_{AB}$ for all $h\in\N_{[1,\overline{h}]}$, and that $\tilde{\Sigma}^h_{AB}$ can be expressed via a quadratic constraint as discussed in \cite{berberich2020combining}. In the following, we will elaborate how to construct $\tilde{\boldsymbol{P}}_d^h$.
\begin{remark} \label{rem:disturbance_bound}
	Assumption \ref{ass:comp_wise_disturbance_bound} directly implies that \eqref{eq:quad_multiplier} with $Q_d = -I$, $S_d=0$ and $R_d = \overline{d}^2NI$ or \eqref{eq:diag_multiplier} are valid classes of multipliers for the overapproximation $\tilde{\mathcal{D}}^1$ (cf. \cite[Subsection II.C]{berberich2020combining}).
\end{remark}

First, note that with Assumption \ref{ass:comp_wise_disturbance_bound} and \eqref{eq:lifted_disturbance}, the 2-norm of the components of $D^h$ can be upper bounded by
\begin{equation} \label{eq:disturbance_bound_via_sing_values}
\begin{aligned}
\lVert d^h(t) \rVert_2 &\le \sum_{i=0}^{h-1} \lVert A_\text{tr}^i \rVert_2 \lVert B_d \rVert_2 \lVert d(t+i) \rVert_2  \\ &\le \sum_{i=0}^{h-1} \sigma_\text{max}(A_\text{tr}^i) \sigma_\text{max}(B_d) \overline{d},
\end{aligned}
\end{equation}
$t\in\{0,\ldots,N-h\}$, $h\in\N_{[1,\overline{h}]}$. Note that $\sigma_{\max}(A_\text{tr}^i)$ is not known since $A_\text{tr}$ is unknown. Therefore, we use in the following the measured data to derive an overapproximation of $\sigma_{\max}(A_\text{tr}^i)$. To this end, consider that if it was known that
\begin{equation} \label{eq:sing_value_upper_bound}
A_\text{tr}^h{A_\text{tr}^{h\top}}\preceq\overline{\sigma}_h^2 I
\end{equation}
for some $\overline{\sigma}_h\ge0$, it would follow directly that $\sigma_\text{max}(A_\text{tr}^h)\le\overline{\sigma}_h$. We can verify a bound like \eqref{eq:sing_value_upper_bound} only in case the same bound holds for all $A^h$ that are consistent with the data. Assume, for now, that a class of multipliers $\tilde{\boldsymbol{P}}_d^h$ such that $\mathcal{D}^h\subseteq\tilde{\mathcal{D}}^h$ was known for a given $h$. It is straightforward to show (cf. \cite[Lemma 2]{berberich2020combining}) that then, if we define
\begin{equation} \label{eq:def_Ph}
\tilde{\boldsymbol{P}}^h_{AB} \coloneqq
\left[\begin{array}{cc} -X^h & 0 \\ -U^h & 0 \\ \hline X_+^h & B_d^h \end{array}\right]
\tilde{\boldsymbol{P}}_d^h
\left[\begin{array}{cc} -X^h & 0 \\ -U^h & 0 \\ \hline X_+^h & B_d^h \end{array}\right]^\top,
\end{equation}
the set $\tilde{\Sigma}^h_{AB}$ can be expressed by the quadratic constraint
\begingroup
\setlength\arraycolsep{2pt}
\begin{equation} \label{eq:cond_consistent_matrices}
\tilde{\Sigma}^h_{AB}\hspace{-1pt}=\hspace{-1pt}\Big\{\hspace{-2pt}\begin{bmatrix} A^h & \underline{B}^h \end{bmatrix}\hspace{-3pt}\Bigm|\hspace{-3pt}
\begin{bmatrix} A^{h\top} \\ \underline{B}^{h\top} \\ I \end{bmatrix}^\top \hspace{-6.5pt}
\tilde{P}^h_{AB}
\hspace{-1pt}
\begin{bmatrix} A^{h\top} \\ \underline{B}^{h\top} \\ I \end{bmatrix}\hspace{-2pt}\succeq\hspace{-2pt}0, \: \forall \tilde{P}^h_{AB}\in\tilde{\boldsymbol{P}}_{AB}^h\hspace{-1pt}\Big\}.
\end{equation}
\endgroup
Then, if we rewrite Condition \eqref{eq:sing_value_upper_bound} as
\begin{equation} \label{eq:cond_sing_value_switched}
\begin{bmatrix} A^{h\top} \\ \underline{B}^{h\top} \\ I \end{bmatrix}^\top
\begin{bmatrix} -I & 0 & 0 \\ 0 & 0 & 0 \\ 0 & 0 & \overline{\sigma}_h^2 I \end{bmatrix}
\begin{bmatrix} A^{h\top} \\ \underline{B}^{h\top} \\ I \end{bmatrix} \succeq 0,
\end{equation}
we may use the S-procedure \cite[Lemma A.1]{scherer2000robust} to derive the condition
\begin{equation} \label{eq:cond_sing_value_switched_S}
\begin{bmatrix} -I & 0 & 0 \\ 0 & 0 & 0 \\ 0 & 0 & \overline{\sigma}_h^2 I \end{bmatrix}-
\tilde{P}^h_{AB} \succeq 0
\end{equation}
which, if it holds for some $\tilde{P}^h_{AB}\in\tilde{\boldsymbol{P}}^h_{AB}$, implies that the quadratic matrix inequality (QMI) \eqref{eq:cond_sing_value_switched} holds for all $\begin{bmatrix} A^h & \underline{B}^h \end{bmatrix}$ which satisfy \eqref{eq:cond_consistent_matrices}. Since $\Sigma^h_{AB}\subseteq\tilde{\Sigma}^h_{AB}$, this provides us with a tractable condition to estimate the maximum singular value directly from noisy data.

In view of \eqref{eq:disturbance_bound_via_sing_values}, we can see that the component-wise bound on $D^h$ depends only on the singular values $\sigma_\text{max}(A^i)$ from $i=1$ to $h-1$. Therefore, it is possible to construct a bound on the lifted disturbance recursively: For a given $h$, first the quadratic disturbance bound $\tilde{\mathcal{D}}^h$ can be used to estimate the maximum singular value $\overline{\sigma}_h$ of $A^h$ via \eqref{eq:cond_sing_value_switched_S}. From this, a component-wise disturbance bound for $D^{h+1}$ can be constructed via \eqref{eq:disturbance_bound_via_sing_values}, which in turn implies that \eqref{eq:quad_multiplier} or \eqref{eq:diag_multiplier} are valid classes of multipliers for the overapproximation $\tilde{\mathcal{D}}^{h+1}$ (cf. Remark \ref{rem:disturbance_bound}). After this, $h$ is increased by one and the procedure is repeated. In Algorithm \ref{algo:disturbance_bound}, this recursive scheme is summarized and formalized.

\begin{algorithm}
	\SetAlgoLined
	\KwIn{$X_+^h$, $X^h$, $U^h$ $\forall h=1,\ldots,\overline{h}-1$, $N$, $B_d$, $\overline{d}$ .}
	\KwOut{$\tilde{\boldsymbol{P}}_d^h$ and $\tilde{\boldsymbol{P}}^h_{AB}$ for all $h=1,\ldots,\overline{h}$.}
	Set $\overline{d}_1 \coloneqq \overline{d}$\; \label{algo:line_init}
	\For{$h=1$ \KwTo $\overline{h}$}{
		\eIf{$h=1$}{$B_d^h=B_d$;}{$B_d^h=I$;}
		Set $N_h\coloneqq N-h+1$, compute $\tilde{\boldsymbol{P}}_d^h$ as in \eqref{eq:quad_multiplier} (with $Q_d = -I$, $S_d=0$ and $R_d = \overline{d}_{h}^2N_{h}I$) or \eqref{eq:diag_multiplier} (with $\overline{d},N$ replaced by $\overline{d}_{h},N_{h}$) and $\tilde{\boldsymbol{P}}^h_{AB}$ as in \eqref{eq:def_Ph}\; \label{algo:line_quad_disturbance_bound}
		Solve the SDP
		$\min_{\overline{\sigma}_h^2,\tilde{P}^h_{AB}\in\tilde{\boldsymbol{P}}^h_{AB}} \; \overline{\sigma}_h^2 \quad \text{s.t.} \quad \eqref{eq:cond_sing_value_switched_S}, \; \overline{\sigma}_h^2\ge 0$\; \label{algo:line_SDP}
		\eIf{SDP admits a solution}{\eIf{h=1}{Set $\overline{\sigma}_h\coloneqq \sqrt{\overline{\sigma}_h^2}$\; \label{algo:line_min3}}{Set $\overline{\sigma}_h\coloneqq \min\{\sqrt{\overline{\sigma}_h^2},\overline{\sigma}_1\overline{\sigma}_{h-1}\}$\; \label{algo:line_min}}
		}{
			Set $\overline{\sigma}_h\coloneqq \overline{\sigma}_1\overline{\sigma}_{h-1}$\; \label{algo:line_min2}}
		Set $\overline{d}_{h+1} \coloneqq (1+\sum_{i=1}^{h} \overline{\sigma}_h) \sigma_\text{max}(B_d) \overline{d}$\;
		\label{algo:line_disturbance_bound}
	}
	
	\caption{Compute lifted disturbance bound and lifted system parametrizations from data.}
	\label{algo:disturbance_bound}
\end{algorithm}

Algorithm \ref{algo:disturbance_bound} takes as inputs the data, the data length $N$, the disturbance matrix $B_d$ and the component-wise bound on the disturbance $\overline{d}$ from Assumption \ref{ass:comp_wise_disturbance_bound}. It returns as outputs the class of multipliers $\tilde{\boldsymbol{P}}_d^h$ for the overapproximations $\tilde{\mathcal{D}}^h$ of the lifted disturbance bound $\mathcal{D}^h$, and the class of multipliers $\tilde{\boldsymbol{P}}^h_{AB}$ parametrizing the set $\tilde{\Sigma}^h_{AB}$. After initialization, the algorithm iterates from $h=1$ to $\overline{h}$ to estimate the respective singular values and disturbance bounds. In Line \ref{algo:line_quad_disturbance_bound}, the class of multipliers $\tilde{\boldsymbol{P}}_d^h$ is defined, which follows immediately from a component-wise disturbance bound $\lVert d^h(t) \rVert_2 \le \overline{d}_h$ (cf. Remark \ref{rem:disturbance_bound}). Since we are interested in possibly small singular values and disturbance bounds, an SDP which minimizes $\overline{\sigma}_h^2$ subject to \eqref{eq:cond_sing_value_switched_S} is solved in Line \ref{algo:line_SDP}. Then, if the SDP admits a solution, an estimate for the maximum singular value of $A^h$ is determined in Lines \ref{algo:line_min3} and \ref{algo:line_min} as the minimum of $\sqrt{\overline{\sigma}_h^2}$, which comes from the SDP, and of $\overline{\sigma}_1\overline{\sigma}_{h-1}$. The latter is a valid upper bound for $\sigma_\text{max}(A^h)$ simply due to the fact that $\sigma_\text{max}(CD)\le\sigma_\text{max}(C)\sigma_\text{max}(D)$ for some matrices $C,D$. If the SDP does not admit a solution, the singular value estimate is simply set to $\overline{\sigma}_1\overline{\sigma}_{h-1}$ in Line \ref{algo:line_min2}. Finally, in Line \ref{algo:line_disturbance_bound}, the component-wise disturbance bound for $D^{h+1}$ is computed based on the singular value estimates and \eqref{eq:disturbance_bound_via_sing_values}.

\begin{remark}
	We remark that for an $h\in\N_{[1,\overline{h}]}$, all members of the set $\{\prod_{i=1}^{h-1} \overline{\sigma}_i^{n_i} \vert \sum_{i=1}^{h-1} in_i = h \}$ are also valid estimates for $\sigma_\text{max}(A_\text{tr}^h)$. Therefore, one might modify Lines \ref{algo:line_min} and \ref{algo:line_min2} of Algorithm \ref{algo:disturbance_bound} to minimize over all members of this set or arbitrary subsets thereof. While this can result in tighter estimates of the singular values and hence, of the disturbance bounds, the corresponding scheme has combinatorial complexity in $\bar{h}$.
\end{remark}

In the following result, we state a simple condition under which Algorithm \ref{algo:disturbance_bound} indeed returns valid disturbance bounds. A proof is omitted, since the statement follows immediately from the preceding discussions.
\begin{lemma} \label{lem:lifted_disturbance_bound}
	Suppose Assumptions \ref{ass:Bd} and \ref{ass:comp_wise_disturbance_bound} are satisfied and there exist $\overline{\sigma}_1^2,\tilde{P}^1_{AB}\in\tilde{\boldsymbol{P}}^1_{AB}$ such that \eqref{eq:cond_sing_value_switched_S} is fulfilled for $h=1$. Then, it holds that $\mathcal{D}^h\subseteq\tilde{\mathcal{D}}^h$ and that $\Sigma^h_{AB}\subseteq\tilde{\Sigma}^h_{AB}$ for all $h\in\N_{[1,\overline{h}]}$, where $\tilde{\boldsymbol{P}}^h_d$ and $\tilde{\boldsymbol{P}}^h_{AB}$, $h\in\N_{[1,\overline{h}]}$, are the outputs of Algorithm \ref{algo:disturbance_bound}.
\end{lemma}

We summarize this subsection: we were able to derive a parametrization of the system matrices $\begin{bmatrix} A^h & \underline{B}^h \end{bmatrix}$ involved in the lifted switched system \eqref{eq:system_switched_lifted} for all $h\in\N_{[1,\overline{h}]}$, using a single state-input trajectory sampled at each time instant. This was achieved by lifting the input, circumventing the need for several experiments with each of the sampling periods of interest. To this end, as an independent contribution, we derived a procedure to estimate the maximum singular value of (monomials of) $A_\text{tr}$ using only measured data.

\subsection{Data-driven stability criteria for analysis and controller design} \label{sec:switched_stability}

Having obtained parametrizations of the lifted system matrices, we may now translate the model-based stability conditions to the data-driven setup. Since the true matrices $\begin{bmatrix} A^h_\text{tr} & \underline{B}^h_\text{tr} \end{bmatrix}$ are unknown, we must verify stability of \eqref{eq:system_switched_lifted} for all ``uncertainties'' $\begin{bmatrix} A^h & \underline{B}^h \end{bmatrix}\in\Sigma^h_{AB}$. For controller design, we also need to make sure that the controller gain matrix $\underline{K}^h$ follows the stacked structure \eqref{eq:lifted_system_matrices}, since only then, the lifted switched system \eqref{eq:system_switched_lifted} is equivalent to the aperiodically sampled system \eqref{eq:system_aper_sampled}. Therefore, we want to achieve stability of the origin of the uncertain lifted switched system
\begin{equation} \label{eq:system_switched_lifted_uncertain}
\begin{aligned}
&x(t_{k+1}) = \left(A^{h_k} + \underline{B}^{h_k} \underline{K}^{h_k}
\right) x(t_k), \\
&h_k\in\N_{[1,\overline{h}]}, \\
&\begin{bmatrix} A^h & \underline{B}^h \end{bmatrix} \in\Sigma^h_{AB}, \; h\in\N_{[1,\overline{h}]}.
\end{aligned}
\end{equation}
A sufficient condition for robust stability of \eqref{eq:system_switched_lifted_uncertain} is the existence of matrices $\mathcal{S}^h=\mathcal{S}^{h\top}\succ 0\in\R^{n\times n}$ and $\mathcal{G}^h\in\R^{n\times n}$, $h\in\N_{[1,\overline{h}]}$, satisfying
\begin{equation} \label{eq:stab_cond_switched_ss}
\begin{bmatrix}
\mathcal{G}^h + \mathcal{G}^{h\top} - \mathcal{S}^h & \star \\
(A^h+\underline{B}^h \underline{K}^h)\mathcal{G}^h & \mathcal{S}^j
\end{bmatrix} \succ 0, \; (h,j)\in \N_{[1,\overline{h}]}^2,
\end{equation}
for all $\begin{bmatrix} A^h & \underline{B}^h \end{bmatrix}\in\Sigma^h_{AB}$, $h\in\N_{[1,\overline{h}]}$ (cf. \cite[Theorem 2]{daafouz2002switched}).

As a step towards verifying stability from data, in the following result we derive an equivalent characterization of Condition \eqref{eq:stab_cond_switched_ss} as a QMI in the variable $\begin{bmatrix} A^h & \underline{B}^h \end{bmatrix}$.
\begin{lemma} \label{lem:stab_switched_equiv_charact}
	Condition \eqref{eq:stab_cond_switched_ss} holds if and only if the conditions $\mathcal{G}^h+\mathcal{G}^{h\top}-\mathcal{S}^h\succ 0$, $h\in\N_{[1,\overline{h}]}$, and
	\begin{equation} \label{eq:stab_cond_switched_ss_mod}
	\begin{bmatrix} A^{h\top} \\ \underline{B}^{h\top} \\ I \end{bmatrix}^\top
	M^{hj}
	\begin{bmatrix} A^{h\top} \\ \underline{B}^{h\top} \\ I \end{bmatrix}\succ 0, \; (h,j)\in \N_{[1,\overline{h}]}^2
	\end{equation}
	hold, where $M^{hj}$ is defined in \eqref{eq:def_Mij}.
\end{lemma}
\begin{figure*}
	\vspace{2pt}
	\renewcommand{\arraystretch}{1.2}
	\begin{equation} \label{eq:def_Mij}
	M^{hj}= \left[
	\begin{array}{c|c}
	M^{hj}_{11} & M^{hj}_{12} \\ \hline
	(M^{hj}_{12})^\top & M^{hj}_{22}
	\end{array} \right]
	\coloneqq
	\left[
	\begin{array}{c|c}
	\begin{matrix} -\mathcal{G}^h (\mathcal{G}^h+\mathcal{G}^{h\top}-\mathcal{S}^h)^{-1} \mathcal{G}^{h\top} & \star \\
	-\underline{K}^h \mathcal{G}^h (\mathcal{G}^h+\mathcal{G}^{h\top}-\mathcal{S}^h)^{-1} \mathcal{G}^{h\top} & -\underline{K}^h \mathcal{G}^h (\mathcal{G}^h+\mathcal{G}^{h\top}-\mathcal{S}^h)^{-1} \mathcal{G}^{h\top} \underline{K}^{h\top}	\end{matrix} & \begin{matrix} \star \\ \star \end{matrix} \\ \hline
	\begin{matrix} 0\hspace{132pt} & \hspace{132pt} 0 \end{matrix} & \mathcal{S}^j \\
	\end{array} \right]
	\end{equation}
	\noindent\makebox[\linewidth]{\rule{\textwidth}{0.4pt}}
\end{figure*}
\begin{proof}
	We apply the Schur complement to \eqref{eq:stab_cond_switched_ss} with respect to the first diagonal element to obtain the equivalent conditions $\mathcal{G}^h+\mathcal{G}^{h\top}-\mathcal{S}^h\succ 0$, $h\in\N_{[1,\overline{h}]}$ and
	\begin{equation*}
	\mathcal{S}^j - (A^h+\underline{B}^h \underline{K}^h)\mathcal{G}^h (\mathcal{G}^h+\mathcal{G}^{h\top}-\mathcal{S}^h)^{-1} \star^\top \succ 0,
	\end{equation*}
	$(h,j)\in \N_{[1,\overline{h}]}^2$. We pull the terms $A^h$, $\underline{B}^h$ and $I$ out of the latter condition to arrive at \eqref{eq:stab_cond_switched_ss_mod}.
\end{proof}

By virtue of Lemma \ref{lem:stab_switched_equiv_charact}, we have now represented the stability condition \eqref{eq:stab_cond_switched_ss} as the QMI \eqref{eq:stab_cond_switched_ss_mod}, which must hold for all $\begin{bmatrix} A^h & \underline{B}^h \end{bmatrix}\in\Sigma^h_{AB}$ in order to conclude robust stability. In order to verify \eqref{eq:stab_cond_switched_ss} for all matrices $\begin{bmatrix} A^h & \underline{B}^h \end{bmatrix}\in\Sigma_{AB}^h$, we now leverage that $\Sigma_{AB}^h\subseteq\tilde{\Sigma}_{AB}^h$ and verify \eqref{eq:stab_cond_switched_ss} for all $\begin{bmatrix} A^h & \underline{B}^h \end{bmatrix}\in\tilde{\Sigma}_{AB}^h$ using the S-procedure \cite[Lemma A.1]{scherer2000robust}. The following result constitutes the main result of this section and allows to verify whether the uncertain system \eqref{eq:system_switched_lifted_uncertain} is robustly stable for a given $\overline{h}$.
\begin{theorem} \label{thm:stab_switched_data_an}
	Suppose Assumptions \ref{ass:Bd} and \ref{ass:comp_wise_disturbance_bound} are satisfied, there exist $\overline{\sigma}_1^2,\tilde{P}^1_{AB}\in\tilde{\boldsymbol{P}}^1_{AB}$ such that \eqref{eq:cond_sing_value_switched_S} is satisfied for $h=1$, and $\tilde{\boldsymbol{P}}^h_{AB}$, $h\in\N_{[1,\overline{h}]}$, are the outputs of Algorithm \ref{algo:disturbance_bound}. Furthermore suppose, given a controller $K$, there exist matrices $\mathcal{S}^h=\mathcal{S}^{h\top}\succ 0\in\R^{n\times n}$, $\mathcal{G}^h\in\R^{n\times n}$, $h\in\N_{[1,\overline{h}]}$, and $\tilde{P}^{hj}_{AB}\in\tilde{\boldsymbol{P}}^h_{AB}$, $(h,j)\in \N_{[1,\overline{h}]}^2$, such that
	\begin{equation} \label{eq:stab_cond_switched_data_an_1}
	\mathcal{G}^h+\mathcal{G}^{h\top}-\mathcal{S}^h \succ 0, \; h \in\N_{[1,\overline{h}]}, \\
	\end{equation}
	and \eqref{eq:stab_cond_switched_data_an_2} are satisfied. Then, the origin of the uncertain switched system \eqref{eq:system_switched_lifted_uncertain} is asymptotically stable for any $\begin{bmatrix} A^h & \underline{B}^h	\end{bmatrix} \in\Sigma^h_{AB}, \; h\in\N_{[1,\overline{h}]}$.
\end{theorem}
\begin{figure*}
	\vspace{2pt}
	\renewcommand{\arraystretch}{1.2}
	\begin{equation} \label{eq:stab_cond_switched_data_an_2}
	\begin{bmatrix}
	\mathcal{G}^h+\mathcal{G}^{h\top}-\mathcal{S}^h & \star & \star & \star \\
	\mathcal{G}^h & 0 & \star & \star \\
	\begin{bmatrix} K^\top & \cdots & K^\top \end{bmatrix}^\top \mathcal{G}^h & 0 & 0 & \star \\
	0 & 0 & 0 & \mathcal{S}^j
	\end{bmatrix}
	-\begin{bmatrix} 0 & \star \\ 0 & \tilde{P}^{hj}_{AB} \end{bmatrix} \succ 0, \; (h,j)\in \N_{[1,\overline{h}]}^2
	\end{equation}
	\noindent\makebox[\linewidth]{\rule{\textwidth}{0.4pt}}
\end{figure*}
\begin{proof}
	We apply the Schur complement to \eqref{eq:stab_cond_switched_data_an_2} with respect to its first diagonal element, which yields that \eqref{eq:stab_cond_switched_data_an_2} is equivalent to
	\begin{equation*}
	M^{hj}-\tilde{P}^{hj}_{AB}\succ 0, \; (h,j)\in \N_{[1,\overline{h}]}^2.
	\end{equation*}
	Next, we use \cite[Lemma A.1]{scherer2000robust} for each $(h,j)\in \N_{[1,\overline{h}]}^2$ to conclude that this condition implies that \eqref{eq:stab_cond_switched_ss_mod} holds for all $\begin{bmatrix} A^h & \underline{B}^h	\end{bmatrix}$ which satisfy \eqref{eq:cond_consistent_matrices}, i.e., for all $\begin{bmatrix} A^h & \underline{B}^h	\end{bmatrix}\in\tilde{\Sigma}^h_{AB}$. Since $\Sigma^h_{AB}\subseteq\tilde{\Sigma}^h_{AB}$, \eqref{eq:stab_cond_switched_ss_mod} also holds for all $\begin{bmatrix} A^h & \underline{B}^h	\end{bmatrix} \in \Sigma^h_{AB}$.
	
	By Lemma \ref{lem:stab_switched_equiv_charact}, the fact that  \eqref{eq:stab_cond_switched_ss_mod} and \eqref{eq:stab_cond_switched_data_an_1} hold for all $\begin{bmatrix} A^h & \underline{B}^h	\end{bmatrix} \in \Sigma^h_{AB}$ implies that \eqref{eq:stab_cond_switched_ss} is satisfied for all $\begin{bmatrix} A^h & \underline{B}^h	\end{bmatrix} \in \Sigma^h_{AB}$ as well. From this, we have stability of the uncertain switched system \eqref{eq:system_switched_lifted_uncertain} for all $\begin{bmatrix} A^h & \underline{B}^h	\end{bmatrix} \in\Sigma^h_{AB}, \; h\in\N_{[1,\overline{h}]}$ by \cite[Theorem 2]{daafouz2002switched}.
\end{proof}

To verify stability for a given $\overline{h}$, one has to follow a two-step procedure: First, one has to invoke Algorithm \ref{algo:disturbance_bound} in order to determine parametrizations of the lifted matrices $\begin{bmatrix} A^h & \underline{B}^h	\end{bmatrix}$, and second, one uses these parametrizations in the stability conditions of Theorem \ref{thm:stab_switched_data_an}.

\begin{remark}
	Note that if one increases the tested sampling interval from $\overline{h}_\text{old}$ to $\overline{h}_\text{new}>\overline{h}_\text{old}$, it is not required to re-run Algorithm \ref{algo:disturbance_bound} from $h=1$. Instead, it is sufficient to iterate over the missing steps $h\in\N_{[\overline{h}_\text{old},\overline{h}_\text{new}-1]}$ due to its recursive nature.
\end{remark}

\begin{remark}
	For a given data length $N$, the highest sampling interval that can be checked via Theorem \ref{thm:stab_switched_data_an} is $\overline{h}=N$.
\end{remark}

Conditions \eqref{eq:stab_cond_switched_data_an_1} and \eqref{eq:stab_cond_switched_data_an_2} in Theorem \ref{thm:stab_switched_data_an} allow for a simultaneous search for matrices $\mathcal{S}^h$, $\mathcal{G}^h$ and $\tilde{\boldsymbol{P}}_{AB}^{hj}$, but not for a stabilizing controller $K$. This is due to the fact that we look for a single controller gain $K$ satisfying \eqref{eq:stab_cond_switched_data_an_2} for each $h$, such that standard convexifying variable transformations are not applicable. For the model-based scenario, an alternative was presented in \cite{xiong2007packet_loss}: It was demonstrated that fixing the matrices $\mathcal{G}^1 = \cdots =\mathcal{G}^{\overline{h}} \eqqcolon \mathcal{G}$ indeed enables the design of a non-switched controller, although this might come with an increase of conservatism. We follow a corresponding approach to allow for controller design in the data-driven case as well.

\begin{corollary} \label{cor:stab_switched_data_co}
	Suppose Assumptions \ref{ass:Bd} and \ref{ass:comp_wise_disturbance_bound} are satisfied, there exist $\overline{\sigma}_1^2,\tilde{P}^1_{AB}\in\tilde{\mathcal{P}}^1_{AB}$ such that \eqref{eq:cond_sing_value_switched_S} is satisfied for $h=1$, and $\tilde{\mathcal{P}}^h_{AB}$, $h\in\N_{[1,\overline{h}]}$, are the outputs of Algorithm \ref{algo:disturbance_bound}. Furthermore, suppose there exist matrices $\mathcal{S}^h=\mathcal{S}^{h\top}\succ 0\in\R^{n\times n}$, $h\in\N_{[1,\overline{h}]}$, $\mathcal{G}\in\R^{n\times n}$, $\mathcal{F}\in\R^{m\times n}$ and $\tilde{P}^{hj}_{AB}\in\tilde{\mathcal{P}}^h_{AB}$, $(h,j)\in \N_{[1,\overline{h}]}^2$, such that
	\begin{equation} \label{eq:stab_cond_switched_data_co_1}
	\mathcal{G}+\mathcal{G}^\top-\mathcal{S}^h \succ 0, \; \forall h \in\N_{[1,\overline{h}]},
	\end{equation}
	and \eqref{eq:stab_cond_switched_data_co_2} are satisfied. Then, the controller $\underline{K}^h \coloneqq \begin{bmatrix} K^\top & \cdots & K^\top \end{bmatrix}^\top$ with $K\coloneqq \mathcal{F}\mathcal{G}^{-1}$ renders the origin of the uncertain switched system \eqref{eq:system_switched_lifted_uncertain} asymptotically stable for any $\begin{bmatrix} A^h & \underline{B}^h	\end{bmatrix} \in\Sigma^h_{AB}, \; h\in\N_{[1,\overline{h}]}$.
\end{corollary}
\begin{figure*}
	\vspace{2pt}
	\renewcommand{\arraystretch}{1.2}
	\begin{equation} \label{eq:stab_cond_switched_data_co_2}
	\begin{bmatrix}
	\mathcal{G}+\mathcal{G}^\top-\mathcal{S}^h & \star & \star & \star \\
	\mathcal{G} & 0 & \star & \star \\
	\begin{bmatrix} \mathcal{F}^\top & \cdots & \mathcal{F}^\top \end{bmatrix}^\top & 0 & 0 & \star \\
	0 & 0 & 0 & \mathcal{S}^j
	\end{bmatrix}
	-\begin{bmatrix} 0 & \star \\ 0 & \tilde{P}^{hj}_{AB} \end{bmatrix} \succ 0, \; (h,j)\in \N_{[1,\overline{h}]}^2
	\end{equation}
	\noindent\makebox[\linewidth]{\rule{\textwidth}{0.4pt}}
\end{figure*}
\begin{proof}
	Since $\mathcal{G}+\mathcal{G}^\top-\mathcal{S}^h \succ 0$ and $\mathcal{S}^h\succ 0$, $\mathcal{G}$ is of full rank such that the transformation $K=\mathcal{F}\mathcal{G}^{-1}$ is uniquely defined. We substitute $\mathcal{F}=K\mathcal{G}$ in \eqref{eq:stab_cond_switched_data_co_2} to obtain \eqref{eq:stab_cond_switched_data_an_2}. We then conclude asymptotic stability via Theorem \ref{thm:stab_switched_data_an}.
\end{proof}

As in Section \ref{sec:IO}, we conclude with the fact that since the true matrices are contained in the set of matrices that can explain the data, robust stability of the uncertain lifted switched system \eqref{eq:system_switched_lifted_uncertain} implies stability of the true switched system \eqref{eq:system_switched}, \eqref{eq:system_switched_lifted} and thus also of the aperiodically sampled system \eqref{eq:system_aper_sampled}. A necessary requirement for this was that $\underline{K}^h$ indeed follows the structure \eqref{eq:lifted_system_matrices}, which we achieved in the controller design result Corollary \ref{cor:stab_switched_data_co}.
\begin{corollary} \label{cor:stab_switched_data_true}
	Suppose the conditions of Theorem \ref{thm:stab_switched_data_an} or Corollary \ref{cor:stab_switched_data_co} are fulfilled. Then the controller $K$ renders the origin of the aperiodically sampled system \eqref{eq:system_aper_sampled} asymptotically stable for all $h_k\in\N_{[1,\overline{h}]}$.
\end{corollary}

Theorem \ref{thm:stab_switched_data_an} provides a tool to check stability of the aperiodically sampled system \eqref{eq:system_aper_sampled} given a controller $K$, and is therefore suitable to tackle Problem \ref{prob:arbitrary_sampling_an}. Since Corollary \ref{cor:stab_switched_data_co} permits a co-search for a stabilizing controller, it may be used to address Problem \ref{prob:arbitrary_sampling_co}.

\section{Numerical Analysis} \label{sec:num}

For a numerical evaluation of our approaches, we consider the example from \cite{zhang2001stability}, exactly discretized with a discretization period of $H=\SI{0.1}{\second}$. The system matrices are then
\begin{equation*}
A_\text{tr} = \begin{bmatrix} 1 & 0.0995 \\ 0 & 0.9900 \end{bmatrix}, \quad B_\text{tr} = \begin{bmatrix} 0.0005 \\ 0.0100 \end{bmatrix},
\end{equation*}
where we rounded off after 4 decimal places. Recall that $A_\text{tr},B_\text{tr}$ are \textit{unknown} to our data-driven approaches. We have $N$ state-input measurements $\{x(t)\}_{t=0}^{N},\{u(t)\}_{t=0}^{N-1}$ at each of the time instants available, where the data-generating input is taken uniformly from $u(t)\in[-1,1]$. Let us assume that the measurements are perturbed by a disturbance $\{\hat{d}(t)\}_{t=0}^{N-1}$, where it holds that $\hat{d}(t)\in[-\overline{d},\overline{d}]$ for some $\overline{d}\ge 0$, which implies that Assumption \ref{ass:comp_wise_disturbance_bound} is fulfilled. Throughout this numerical example, we describe this component-wise bound by the means of diagonal multipliers \eqref{eq:diag_multiplier} in Assumption \ref{ass:disturbance_bound}. In addition, suppose it is known that the disturbance acts only on the first state, which we may incorporate by setting $B_d = \begin{bmatrix} 0.01 & 0 \end{bmatrix}^\top$. Note that with this choice of $B_d$, Assumption \ref{ass:Bd} is fulfilled and a certain $\overline{d}$ corresponds to an input-to-noise ratio of approximately $1/\overline{d}$. Furthermore, in all of the tested cases, there existed $\overline{\sigma}_1^2,\tilde{P}^1_{AB}\in\tilde{\mathcal{P}}^1_{AB}$ s.t. \eqref{eq:cond_sing_value_switched_S} was fulfilled for $h=1$, which is a precondition for Theorem \ref{thm:stab_switched_data_an} and Corollary \ref{cor:stab_switched_data_co}. The numerical results were obtained using MatlabR2019b, YALMIP~\cite{YALMIP} and Mosek~\cite{MOSEK15}.

\subsection{Comparison of robust input/output and switched systems approach}

\subsubsection{Comparison of different sizes $N$ of the data set}

\begin{figure}
	\centering
	\begin{tikzpicture}
	
	\begin{semilogxaxis}[%
		width=\columnwidth,
		height=5.3cm,
		xmin=0.0005,
		xmax=1,
		extra x ticks = {0.0005},
		extra x tick labels = {$0$},
		xlabel={Disturbance level $\overline{d}$},
		ymin=0,
		ymax=22,
		ylabel={$\overline{h}_{\text{MSI}}$},
		yminorticks=true,
		axis background/.style={fill=white},
		legend style={at={(0.97,0.97)}, anchor=north east, legend cell align=left, align=left, draw=white!15!black, legend columns=3}
		]
		
		\addplot [color=black,mark=square,,draw=none]
		table[row sep=crcr]{%
			0.0005	12\\
		};
		\addlegendentry{$N=5$}
		
		\addplot [color=black,mark=o,,draw=none]
		table[row sep=crcr]{%
			0.0005	12\\
		};
		\addlegendentry{$N=50$}
		
		\addplot [color=black,mark=star,,draw=none]
		table[row sep=crcr]{%
			0.0005	12\\
		};
		\addlegendentry{$N=500$}
		
		\addplot [color=istorange,mark=square,forget plot]
		table[row sep=crcr]{%
			0.0005	12\\
			0.001	12\\
			0.002	11\\
			0.005	11\\
		};
		
		\addplot [color=istblue,mark=square,forget plot,dashed]
		table[row sep=crcr]{%
			0.0005	12\\
			0.001	12\\
			0.002	12\\
			0.005	11\\
		};
		
		\addplot [color=istblue,mark=o,forget plot,dashed]
		table[row sep=crcr]{%
			0.0005	12\\
			0.001	12\\
			0.002	12\\
			0.005	12\\
			0.01	12\\
			0.02	12\\
			0.05	12\\
			0.1		12\\
			0.2		11\\
			0.5		11\\
		};
		
		\addplot [color=istblue,mark=star,forget plot,dashed]
		table[row sep=crcr]{%
			0.0005	12\\
			0.001	12\\
			0.002	12\\
			0.005	12\\
			0.01	12\\
			0.02	12\\
			0.05	12\\
			0.1		12\\
			0.2		12\\
			0.5		12\\
			1		12\\
		};
		
		\addplot [color=istorange,mark=o,forget plot]
		table[row sep=crcr]{%
			0.0005	12\\
			0.001	12\\
			0.002	12\\
			0.005	12\\
			0.01	12\\
			0.02	12\\
			0.05	12\\
			0.1		12\\
			0.2		11\\
		};
		
		\addplot [color=istorange,mark=star,forget plot]
		table[row sep=crcr]{%
			0.0005	12\\
			0.001	12\\
			0.002	12\\
			0.005	12\\
			0.01	12\\
			0.02	12\\
			0.05	12\\
			0.1		12\\
			0.2		12\\
			0.5		12\\
			1		11\\
		};
		
		\addplot [color=istgreen,mark=square,forget plot]
		table[row sep=crcr]{%
			0.0005	2\\
			0.001	2\\
			0.002	2\\
			0.005	1\\
		};
		
		\addplot [color=istgreen,mark=o,forget plot]
		table[row sep=crcr]{%
			0.0005	17\\
			0.001	17\\
			0.002	16\\
			0.005	16\\
			0.01	15\\
			0.02	15\\
			0.05	12\\
			0.1		6\\
			0.2		1\\
		};
		
		\addplot [color=istgreen,mark=star,forget plot]
		table[row sep=crcr]{%
			0.0005	17\\
			0.001	17\\
			0.002	17\\
			0.005	16\\
			0.01	16\\
			0.02	15\\
			0.05	14\\
			0.1		12\\
			0.2		10\\
			0.5		2\\
			1		1\\
		};

	\end{semilogxaxis}
\end{tikzpicture}
	\caption{MSI bounds with $K =-\begin{bmatrix} 3.75 & 11.5 \end{bmatrix}$, computed via robust input/output (Theorem~\ref{thm:stab_IO_data} \textcolor{istorange}{---}) and switched systems approach (Theorem \ref{thm:stab_switched_data_an} \textcolor{istgreen}{---}), and two-step procedure (set membership estimation \& Theorem \ref{thm:stab_IO_ss} \textcolor{istblue}{- -}), for three different data lengths and various disturbance levels $\overline{d}$. If there is no marker, the stability conditions were not feasible for $\overline{h}=1$.}
	\label{fig:MSI_different_N}
\end{figure}
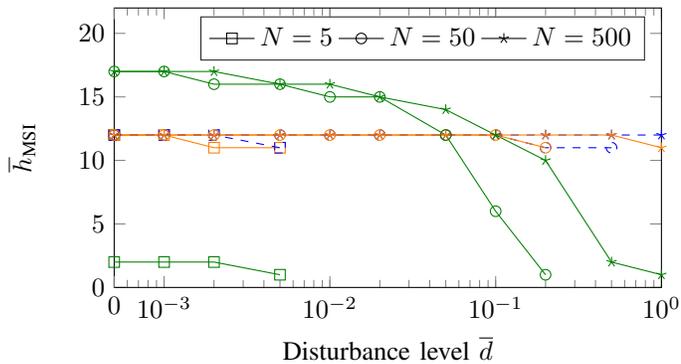

First, let us consider Problem \ref{prob:arbitrary_sampling_an} for the given controller $K =-\begin{bmatrix} 3.75 & 11.5 \end{bmatrix}$, the same that was considered in \cite{fridman2010refined,seuret2018wirtinger,nagshtabrizi2008sampled,mirkin2007some}. For this special scenario, the MSI can be computed exactly if model knowledge is available. It was shown in \cite{nagshtabrizi2008sampled} that it amounts to \SI{1.7}{\second}, which corresponds to $\overline{h}=17$ with our chosen discretization period. We have presented two possibilities to estimate lower bounds on the MSI directly from data, either via Theorem \ref{thm:stab_IO_data} using the robust input/output approach or via Theorem \ref{thm:stab_switched_data_an} using the switched systems approach. Both approaches were tested with three different data lengths $N=5$, $N=50$ and $N=500$, and different disturbance levels to investigate their effect on the MSI bounds. The results can be found in Figure \ref{fig:MSI_different_N}.

As a first observation, we see that the guaranteed MSI bounds decrease with increasing disturbance level for all three data lengths and for both the robust input/output and the switched systems approach. For the robust input/output approach, for all tested data lengths, the MSI estimate amounts to $\overline{h}_{\text{MSI}}=12$ when $\overline{d}=0$. We also computed the MSI using the model-based robust input/output conditions in Theorem \ref{thm:stab_IO_ss} and obtained the same value, i.e., $\overline{h}_\text{MSI}=12$. With increasing disturbance level $\overline{d}$, initially, the MSI estimate barely decreases.  However, at a certain $\overline{d}$ , the stability conditions become infeasible, i.e., we do not find a solution even for $\overline{h}=1$. We observe that for larger data sets, the minimal disturbance level for which we have infeasibility increases.

In contrast to the robust input/output approach, with the switched systems approach we observe that the MSI estimates improve significantly if more data are available. While the estimates are lower than those of the robust input/output approach for $N=5$, they are significantly higher if $N=50$ or $N=500$ and the disturbance level is small. In some cases, even the true MSI of 17 is recovered despite the data-driven setup and non-zero noise.

When estimating the MSI from data, there are in general two main sources of conservatism: The first one being that of the data-driven parametrization, and the second one being that of the (model-based) stability conditions. In the robust input/output approach, the conservatism of the data-driven parametrization in this example seems to be rather low, since the estimated MSIs are higher than for the switched systems approach especially when the data set is small and the disturbance level is high. On the other hand, the conservatism induced by comprehending the delay operator as a bounded disturbance seems to be quite high in this example. Even if the model was known exactly we find $\overline{h}_\text{MSI}=12$, compared to the true MSI of 17. As discussed in Section \ref{sec:switched}, the data-driven parametrizations of the lifted matrices in the switched systems approach introduce conservatism since they rely on an overapproximation of the set of compatible matrices. As a result, the latter is outperformed by the robust input/output approach if the introduced conservatism is considerable, e.g., when the noise level is high and/or few data are available. The model-based switched stability conditions, in contrast, seem to come with little conservatism compared to the ones from the robust input/output approach in this example. As a result, the MSI bounds with the switched systems approach are very tight and are much improved over the ones from the robust input/output approach if there is little noise and sufficiently informative data available.

\subsubsection{Controller design}

\begin{figure}
	\centering
	\begin{tikzpicture}
	
	\begin{semilogxaxis}[%
		width=\columnwidth,
		height=4cm,
		xmin=0.0005,
		xmax=1,
		extra x ticks = {0.0005},
		extra x tick labels = {$0$},
		xlabel={Disturbance level $\overline{d}$},
		ymin=0,
		ymax=65,
		ylabel={$\overline{h}_{\text{MSI}}$},
		yminorticks=true,
		axis background/.style={fill=white},
		]

		\addplot [color=istorange,mark=star,forget plot]
		table[row sep=crcr]{%
			0.0005	62\\
			0.001	61\\
			0.002	54\\
			0.005	36\\
			0.01	29\\
			0.02	23\\
			0.05	15\\
			0.1		11\\
			0.2		6\\
			0.5		2\\
		};
		
		\addplot [color=istgreen,mark=star,forget plot]
		table[row sep=crcr]{%
			0.0005	24\\
			0.001	24\\
			0.002	23\\
			0.005	22\\
			0.01	20\\
			0.02	18\\
			0.05	15\\
			0.1		9\\
			0.2		4\\
			0.5		1\\
		};

	\end{semilogxaxis}
\end{tikzpicture}
	\caption{MSI bounds with controller $K$ computed via robust input/output (Corollary~\ref{cor:stab_IO_data} \textcolor{istorange}{---}) and switched systems approach (Corollary~\ref{cor:stab_switched_data_co} \textcolor{istgreen}{---}), for $N=50$ and various disturbance levels $\overline{d}$.}
	\label{fig:MSI_K_free}
\end{figure}
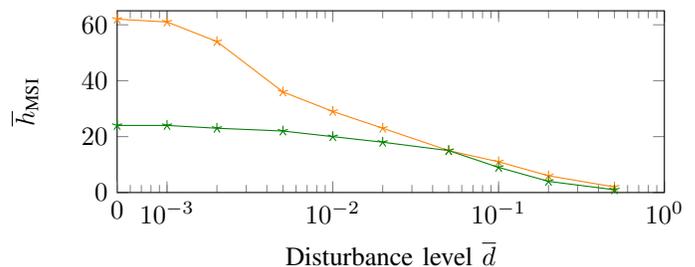

Second, we turn our attention to Problem \ref{prob:arbitrary_sampling_co}, where the aim is to optimize for a controller that gives a possibly high MSI bound. In the robust input/output approach, we can do so by virtue of Corollary \ref{cor:stab_IO_data}, while in the switched systems approach, we may use Corollary \ref{cor:stab_switched_data_co} for this task. The corresponding MSI bounds with $N=50$ and different $\overline{d}$ are presented in Figure \ref{fig:MSI_K_free}.

We observe that for both approaches, the computed MSI bounds with $K$ as an optimization variable are increased compared to $K =-\begin{bmatrix} 3.75 & 11.5 \end{bmatrix}$, especially so for the robust input/output approach. Now, the robust input/output approach yields larger MSI estimates than the switched systems approach for any of the considered disturbance levels.

\subsubsection{Comparison of complexity}

\begin{table}
	\caption{Number of decision variables and constraints in the robust input/output (Theorems \ref{thm:stab_IO_ss} and \ref{thm:stab_IO_data}) and switched systems approach (\cite[Theorem 2]{daafouz2002switched} and Theorem \ref{thm:stab_switched_data_an}), and two-step procedure (set membership estimation \& Theorem \ref{thm:stab_IO_ss}).}
	\begin{tabular}{m{2.6cm}|cc}
		& \# decision variables & \# constraints \\ \hline
		Theorem \ref{thm:stab_IO_ss} & $2n^2$ & $3$ \\
		Theorem \ref{thm:stab_IO_data} & $2n^2 + c_d(N)$ & $3 + c_d(N)$ \\
		\cite[Theorem 2]{daafouz2002switched} & $2\overline{h} n^2$ & $2\overline{h} + \overline{h}^2$ \\
		Algorithm \ref{algo:disturbance_bound} & $\overline{h}\times(1 + c_d(N))$ & $\overline{h}\times(2 + c_d(N))$ \\
		Theorem \ref{thm:stab_switched_data_an} & $2\overline{h} n^2 + \overline{h}^2 c_d(N)$ & $2\overline{h}+\overline{h}^2 + \overline{h}^2c_d(N)$ \\
		Set membership estimation \& Theorem \ref{thm:stab_IO_ss} & $2n^2$ & $2+2^{n(n+m)}$ \\
	\end{tabular}
	\label{tab:complexity}
\end{table}

\begin{figure}
	\centering
	\begin{tikzpicture}
	\begin{axis}[
		width=\columnwidth,
		height=5cm,
		ylabel={Computation time in \SI{}{\second}},
		xlabel = {$\overline{h}$},
		enlargelimits=0.15,
		legend style={at={(0.03,0.97)},
			anchor=north west},
		ybar stacked,
		bar width=7pt,
		]
		\addplot[draw=istorange,fill=istorange!30!white]
		coordinates {(1.6,0.3516) (2.4,0)
			(3.6,0.3444) (4.4,0) (5.6,0.3362) (6.4,0) (7.6,0.3527) (8.4,0) (9.6,0.3450) (10.4,0) (11.6,0.3613) (12.4,0)};
		
		\addplot[draw=istred,fill=istred!30!white]
		coordinates {(1.6,0) (2.4,0.58) (3.6,0) (4.4,1.02) (5.6,0)
			(6.4,1.48) (7.6,0) (8.4,1.94) (9.6,0) (10.4,2.63) (11.6,0) (12.4,3.02)};
		\addplot[draw=istgreen,fill=istgreen!30!white]
		coordinates {(1.6,0) (2.4,0.19) (3.6,0) (4.4,0.21) (5.6,0)
			(6.4,0.29) (7.6,0) (8.4,0.42) (9.6,0) (10.4,0.71) (11.6,0) (12.4,1.05)};
	
		\legend{Theorem \ref{thm:stab_IO_data},Algorithm \ref{algo:disturbance_bound},Theorem \ref{thm:stab_switched_data_an}}
	\end{axis}
\end{tikzpicture}
	\caption{Computation times to check the stability conditions in for the robust input/output (Theorem~\ref{thm:stab_IO_data}) and switched systems approach (Algorithm~\ref{algo:disturbance_bound} and Theorem~\ref{thm:stab_switched_data_an}), for $N=50$, disturbance level $\overline{d}=0.01$ and different $\overline{h}$.}
	\label{fig:computation_time}
\end{figure}
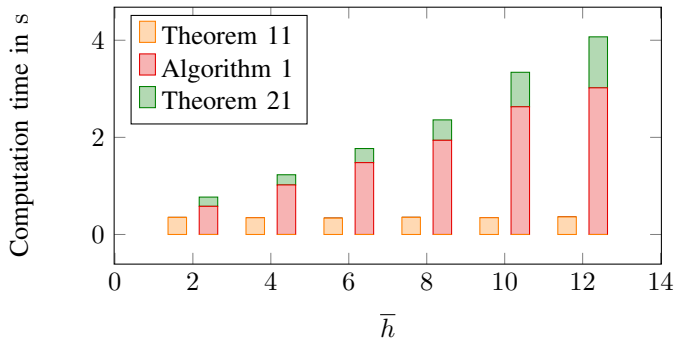

Here, we will briefly compare the computational complexity of the robust input/output and switched systems approach. Table \ref{tab:complexity} gives an overview over the number of decision variables and the number of constraints involved in the respective stability conditions for analysis, i.e., Theorem \ref{thm:stab_IO_ss} and \cite[Theorem 2]{daafouz2002switched} in the model-based case, and Theorems \ref{thm:stab_IO_data} and \ref{thm:stab_switched_data_an}, and Algorithm \ref{algo:disturbance_bound} in the data-driven case. In Figure \ref{fig:computation_time}, the computation times required to check the stability conditions in Theorems \ref{thm:stab_IO_data} and \ref{thm:stab_switched_data_an}, and to run Algorithm \ref{algo:disturbance_bound} are depicted for various $\overline{h}$. All experiments were run on an Intel Core i7-105110U with 1.80 GHz.

Comparing first the number of decision variables and constraints in the data-driven stability conditions of the robust input/output and switched systems approach in Table \ref{tab:complexity}, we note that they are independent of $\overline{h}$ for the former, whereas they grow quadratically with $\overline{h}$ for the latter. In addition, one needs to run Algorithm \ref{algo:disturbance_bound} in the switched systems approach, which requires solving $\overline{h}$ SDPs whose complexity is independent of $\overline{h}$. This is also reflected in the required computation times in Figure \ref{fig:computation_time}: They are independent of $\overline{h}$ for the robust input/output approach and amount to approximately \SI{0.35}{\second} for all experiments. By contrast, we can observe that the required time for the switched systems approach increases with $\overline{h}$.

Second, we compare the complexity of the data-driven stability conditions Theorems \ref{thm:stab_IO_data} and \ref{thm:stab_switched_data_an} with their model-based counterparts Theorem \ref{thm:stab_IO_ss} and \cite[Theorem 2]{daafouz2002switched}. We note that the number of decision variables and constraints is the same for the data-driven conditions as for the model-based ones, with the exception that the former additionally include the number of decision variables $c_d(N)$ involved in the chosen disturbance multiplier. As a result, the increase in complexity with increasing system dimension is the same for the data-driven and model-based conditions.

Third, we discuss the influence of data length on complexity and performance of the proposed data-driven approaches. To this end, assume that the disturbance satisfies a component-wise bound as in Assumption \ref{ass:comp_wise_disturbance_bound}. Then, one could choose either a quadratic or a diagonal disturbance multiplier (see Remark \ref{rem:disturbance_bound}). With a quadratic multiplier, one has $c_d(N)=1$ such that the complexity of the proposed data-driven approaches is independent of the data length. On the other hand, there is no guarantee that more data will shrink the set of compatible matrices and thereby lead to a larger MSI estimate. Quite the contrary, in \cite{berberich2020combining,bisoffi2021trade} it was shown that adding more data can even decrease performance in this case. In contrast, a diagonal multiplier provides the guarantee that more data will never decrease performance \cite{berberich2020combining,bisoffi2021trade}. However, this comes at the price that complexity of the stability conditions now increases linearly with the data length as $c_d(N)=N$.

\subsection{Comparison with alternative methods}

\subsubsection{Comparison with \cite{rueda2021delay}}

As discussed in the introduction, the results of \cite{rueda2021delay} on data-driven control of time-delay systems permit an analysis of aperiodically sampled systems as well. The prerequisites of \cite[Theorem 4.2]{rueda2021delay} were fulfilled in the noise-free case $\overline{d}=0$ if $N\ge 5$. For any such data length, this result yielded an MSI estimate of $\overline{h}_\text{MSI}=4$ with the controller $K =-\begin{bmatrix} 3.75 & 11.5 \end{bmatrix}$.

In most cases, the data-driven methods proposed in this paper yield a considerably better estimate of the MSI. This might be due to the fact that since the latter work considers arbitrary bounded delays, it comes with a certain degree of conservatism when applied to aperiodically sampled systems (cf. Subsections \ref{sec:IO_setup} and \ref{sec:IO_delay_operator}).

\subsubsection{Comparison with two-step procedure}

To leverage data for an estimation of the MSI, an alternative to the approaches presented in this paper is to perform a two-step procedure of system identification and subsequently checking the model-based stability conditions in \cite{seuret2018wirtinger,hetel2011discrete,xiong2007packet_loss,yu2004dropout} or in Theorem \ref{thm:stab_IO_ss}. In this subsection, we will briefly compare this two-step procedure to our data-driven approaches numerically. We will consider set membership estimation \cite{milanese1991optimal,belforte1990parameter} for the system identification step, since it delivers guaranteed error bounds despite the presence of disturbances. In particular, we follow a straightforward approach and estimate a polytope in which the true matrices are guaranteed to lie, subsequently overbound this polytope by a hypercube and check the stability conditions in Theorem \ref{thm:stab_IO_ss} with $K =-\begin{bmatrix} 3.75 & 11.5 \end{bmatrix}$ for each of its vertices. Note that this procedure provides a guaranteed lower bound on the MSI. The results of this method can be found in Figure \ref{fig:MSI_different_N}, next to the results of the proposed data-driven method based on the robust-input output approach, namely Theorem \ref{thm:stab_IO_data}. The number of involved decision variables and constraints for both options are summarized in Table \ref{tab:complexity}.

It can be recognized that, for a small set of noise levels, the two-step procedure gives slightly better results than the proposed data-driven approach. However, the number of involved constraints grows exponentially with the square of the system dimension. This is due to the fact that one constraint is added for each vertex of the matrix hypercube, whose dimension is equal to the number of coefficients in the matrices $\begin{bmatrix} A & B \end{bmatrix}$, namely $n(n+m)$.  For this reason, the stability conditions in this two-step procedure become practically intractable even for modest system dimensions. In contrast, and as already recognized above, the proposed data-driven approaches grow on the same order as their model-based counterparts, namely quadratically with the system dimension.

\section{Summary and Outlook} \label{sec:summary}

In this article, we approached a problem at the intersection of sampled-data control and data-driven control: We developed tools to lower bound the maximum sampling interval of a system under aperiodic sampling, requiring no model knowledge and using only a measured trajectory. Thereby, we considered both analysis and controller design directly from the available data, which may be of finite length and subject to noise. In particular, we presented two distinct approaches to achieve these goals, the first taking a robust control perspective and the second a switched systems perspective on the aperiodically sampled system. The former, dubbed robust input/output approach, comes with a lower computational complexity and is able to produce decent estimates of the maximum sampling interval even if the data set is small and the noise level is high. The latter, dubbed switched systems approach, can yield very tight estimations of the MSI especially when the data set is sufficiently large and there is little noise. The validity of both approaches was illustrated with a numerical example.

In the continuous-time formulation of the robust input/output approach, a passivity-like property of the delay operator was established in addition to a bound on the $\ell_2$ gain. Future work could investigate this topic in the discrete-time case as well, since it was shown in \cite{fujioka2009IQC} that incorporating this property might lead to greatly improved estimates of the MSI. In addition, it might be possible to improve the $\ell_2$ gain estimate of the delay operator even further. In the switched systems approach, future work could focus on whether and how less conservative parametrizations of the lifted matrices could be obtained. A starting point might be to construct tighter overapproximations of the lifted disturbance, e.g., by trying to incorporate information on its structure (cf. Equation \eqref{eq:lifted_disturbance}). A further line of future research could investigate how to incorporate a performance objective into the switched systems approach. Finally, future work could extend data-driven analysis of the MSI to more general system classes, e.g., polynomial systems \cite{martin2021dissipativity,guo2021polynomial} or general nonlinear systems \cite{martin2021datadriven}.

Since the importance of data as well as of cyber-physical, embedded and networked control systems continues to grow, combining concepts from data-driven control and sampled-data control is a highly relevant research direction. The data-driven analysis of aperiodically sampled systems, as presented in this work, may contribute to this emerging field by providing a novel approach to model and analyze a great variety of problems at the intersection of data-driven and sampled-data control, such as learning event-triggered control \cite{sedghi2020machine}, learning unknown channel conditions \cite{gatsis2018sample}, or data-driven network access scheduling \cite{leong2020deep}.

\bibliographystyle{IEEEtran}   
\bibliography{Literature}

\appendix
\section{Appendix}
\subsection{Proof of Lemma \ref{lem:L2_gain}} \label{app:proof_lem_L2}

As done in the continuous-time case \cite{mirkin2007some}, we handle the operator $\Delta$ in the lifted domain. For a signal $g\in\ell_{2e}^n$, the corresponding \emph{lifted signal} is defined as
\begin{equation*}
\underline{g} \coloneqq \left\{ \begin{bmatrix} g(0) \\ \vdots \\ g(t_1-1) \end{bmatrix},\begin{bmatrix} g(t_1) \\ \vdots \\ g(t_2-1) \end{bmatrix}, \ldots
\right\}.
\end{equation*}
Further, for a $T\in\N_0$, we define the \textit{lifted truncated signal} corresponding to $g_T$ as
\begin{equation*}
\underline{g}_T \coloneqq \left\{ \begin{bmatrix} g(0) \\ \vdots \\ g(t_1-1) \end{bmatrix}, \ldots,
\begin{bmatrix} g(t_K) \\ \vdots \\ g(T) \\ 0 \end{bmatrix},\begin{bmatrix} 0 \\ \vdots \\ 0 \end{bmatrix},\ldots
\right\},
\end{equation*}
where $K\coloneqq\max\{k\;\vert\; t_k\le T\}$. As the $t_k$ are not equidistant in time, the time axis in the lifted signal is split non-uniformly as well. Nonetheless, the $\ell_2$ norm of the original signal is preserved in the lifted domain $\lVert \underline{g}_T \rVert_{\ell_2}^2 = \lVert g_T \rVert_{\ell_2}^2$ \cite{chen1995optimal}. Further, we introduce a lifted operator $\underline{\Delta}$ mapping $\underline{y}\to\underline{e}$ via
\begin{equation*}
\underline{e}(k)=(\underline{\Delta}\underline{y})(k)\coloneqq C_k \underline{y}(k), \; k\in\N_0
\end{equation*}
with
\begin{equation*}
C_{k} \coloneqq \begin{bmatrix} 0 & \cdots & 0 & 0 \\ I & \ddots & \vdots & \vdots \\ \vdots & \ddots & 0 & 0 \\ I & \cdots & I & 0 \end{bmatrix}\in\R^{h_k\times h_k}.
\end{equation*}

Now, let us rewrite $\Delta$ using $\tau(t) = t- t_k$ as
\begin{equation*} \label{eq:def_delta_alt}
(\Delta y)(t) \coloneqq \sum_{i= t_k}^{t-1}  y(i), \; t\in\N_{[t_k,t_{k}+h_k-1]}, \; k\in\N_{0},
\end{equation*}
from which it is easy to recognize that $\underline{\Delta}\underline{y}_T=\underline{(\Delta y)}_T$ for any $T\in\N_0$. Since lifting preserves the signal norms, it clearly holds that $\lVert\underline{\Delta}\underline{y}_T\rVert_{\ell_2}=\lVert\Delta y_T\rVert_{\ell_2}$. Since both $\underline{\Delta}$ and $\Delta$ are causal and $\lVert \underline{y}_T \rVert_{\ell_2}^2 = \lVert y_T \rVert_{\ell_2}^2$, we conclude that the $\ell_2$ gain of $\underline{\Delta}$ is equal to that of $\Delta$.

The crucial property of $\underline{\Delta}$ is that it is static, i.e., $\underline{e}(k)$ depends on $\underline{y}(k)$ only. Furthermore, it merely sums the inputs $y(t)$ in between sampling instants. As a result, its $\ell_2$ gain is the maximum $\ell_2$ gain of the summation operator over the intervals $[0,h_k-1]$, $k\in\N_0$, which is clearly attained in the longest possible interval $\overline{h}-1$. To summarize, the $\ell_2$ gain of $\underline{\Delta}$ is given by that of $D:\ell_{2e}^{n}[0,\overline{h}-1]\to\ell_{2e}^{n}[0,\overline{h}-1]$, $y\mapsto Dy$, $(Dy)(t)\coloneqq\sum_{i=0}^{t-1}y(i)$, where $\ell_{2}[0,\overline{h}-1]$ denotes the space of bounded signals of length $\overline{h}$.

We will now bound $\lVert D\rVert_{\ell_2}=\sup_{y\in\ell_{2e}^n,y\neq 0} \frac{\lVert Dy \rVert_{\ell_2}}{\lVert y \rVert_{\ell_2}}$ explicitly. Let us first take a look at the square of the fraction
\begin{equation}
\frac{\lVert Dy \rVert^2_{\ell_2}}{\lVert y \rVert^2_{\ell_2}} = \frac{\sum_{t=0}^{\overline{h}-1} \left(\sum_{i=0}^{t-1}y(i)^\top\right)\cdot \left(\sum_{i=0}^{t-1}y(i)\right)}{\sum_{t=0}^{\overline{h}-1}y(t)^\top y(t)}. \label{eq:L2_fraction}
\end{equation}
Next, we turn our attention to the product of sums in the numerator, which can be rewritten and upper bounded for all $t\in\N_{[0,\overline{h}-1]}$ as follows (denoting the maximum eigenvalue of a symmetric matrix $A\in\R^{n\times n}$ by $\lambda_\text{max}(A)$):
\begin{align}
&\left(\textstyle\sum_{i=0}^{t-1}y(i)^\top\right) \cdot \left(\textstyle\sum_{i=0}^{t-1}y(i)\right) \nonumber \\
&=	\begin{bmatrix} y(0) \\ \vdots \\ y(t-1) \\ y(t)	\end{bmatrix}^\top \Bigg(
\underbrace{\begin{bmatrix} 1 & \cdots & 1 & 0 \\
	\vdots & \ddots & \vdots & \vdots \\
	1 & \cdots & 1 & 0 \\
	0 & \cdots & 0 & 0 	\end{bmatrix}}_{\eqqcolon E_t}\otimes I\Bigg)
\begin{bmatrix} y(0) \\ \vdots \\ y(t-1) \\ y(t) \end{bmatrix} \nonumber \\
&\le \lambda_\text{max}(E_t\otimes I) \big\lVert \hspace{-1.1pt} \left[ y(0) \; \cdots \; y(t) \right]\hspace{-1.1pt} \big\rVert_2^2 = t \sum_{i=0}^{t}y(i)^\top y(i). \label{eq:L2_estimate}
\end{align}
Inequality \eqref{eq:L2_estimate} follows directly from the facts that $E_t\otimes I$ is symmetric and that $\lambda_\text{max}(E_t\otimes I)=\lambda_\text{max}(E_t)\lambda_\text{max}(I)=t$. Since it clearly holds that $\lambda_\text{max}(I)=1$, it remains to prove that $\lambda_\text{max}(E_t)=t$: First, we check that $t$ is indeed an eigenvalue of $E_t$ with the corresponding eigenvector $[1 \; \cdots \; 1 \; 0	]^\top$. Second, we note that $E_t$ is positive semidefinite, since it can be factorized as $E_t = [1 \; \cdots \; 1 \; 0]^\top [*]$, i.e., all eigenvalues are real and greater than or equal to zero. Lastly, we note that $\text{tr}(E_t) = t = \sum_i \lambda_i(E_t)$ and combine this finding with the first two to conclude that $\lambda_\text{max}(E_t)=t$.

We plug the estimate \eqref{eq:L2_estimate} into \eqref{eq:L2_fraction} to obtain
\begin{align*}
\frac{\lVert Dy \rVert^2_{\ell_2}}{\lVert y \rVert^2_{\ell_2}}
&\le	\frac{\sum_{t=0}^{\overline{h}-1} t \sum_{i=0}^{t}y(i)^\top y(i)}{\sum_{t=0}^{\overline{h}-1}y(t)^\top y(t)} \\
&\le \frac{\sum_{t=0}^{\overline{h}-1} t \sum_{i=0}^{\overline{h}-1}y(i)^\top y(i)}{\sum_{t=0}^{\overline{h}-1}y(t)^\top y(t)}
= \sum_{t=0}^{\overline{h}-1} t = \frac{\overline{h}}{2}(\overline{h}-1),
\end{align*}
which concludes the proof.

\begin{IEEEbiography}[{\includegraphics[width=1in,height=1.25in,clip,keepaspectratio]{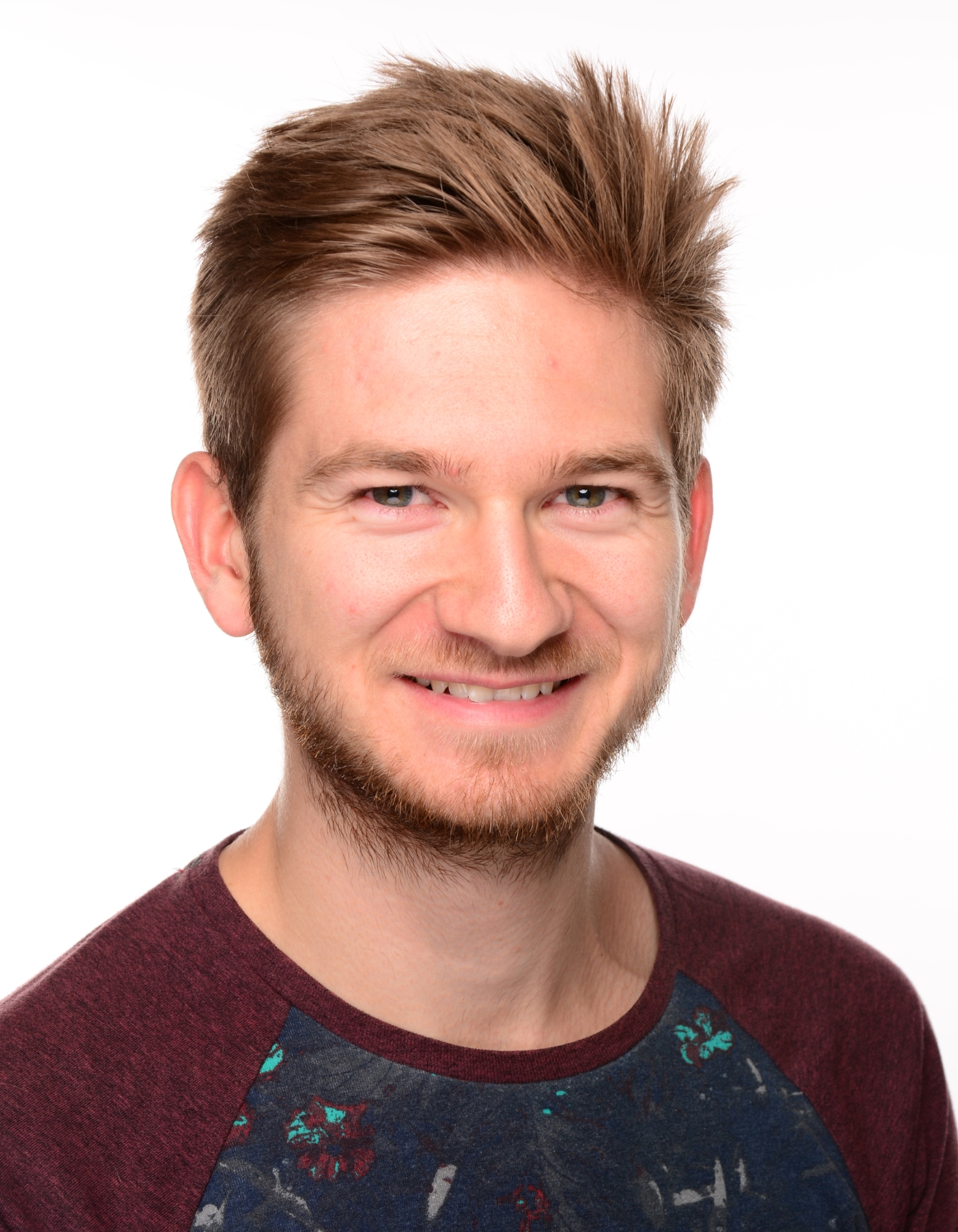}}]{Stefan Wildhagen} received the Master’s degree in Engineering Cybernetics from the University of Stuttgart, Germany, in 2018. He has since been a doctoral student at the Institute for Systems	Theory and Automatic Control under supervision of Prof. Allg\"ower and a member of the Graduate School Simulation Technology at the University of Stuttgart. His research interests are in the area of Networked Control Systems, with a focus on optimization-based scheduling and control as well as on data-driven methods.
\end{IEEEbiography}

\begin{IEEEbiography}[{\includegraphics[width=1in,height=1.25in,clip,keepaspectratio]{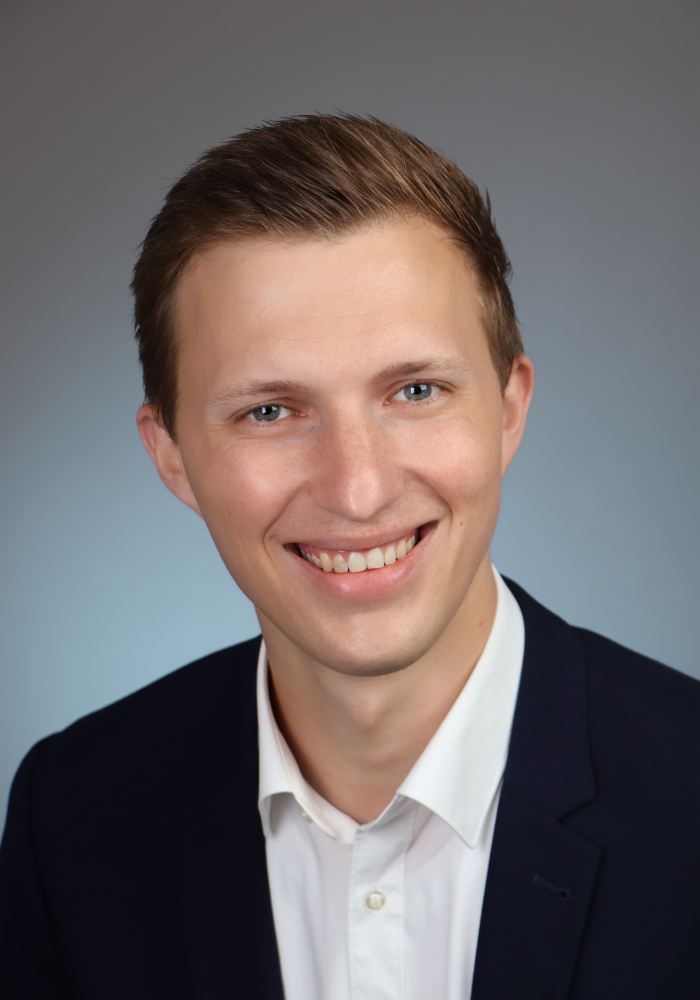}}]{Julian Berberich} received the Master’s degree in Engineering Cybernetics from the University of Stuttgart, Germany, in 2018. Since 2018, he has been	a Ph.D. student at the Institute for Systems Theory and Automatic Control under supervision of Prof.	Frank Allg\"ower and a member of the International Max-Planck Research School (IMPRS) at the University of Stuttgart. He has received the Outstanding Student Paper Award at the 59th Conference on Decision and Control in 2020. His research interests are in the area of data-driven analysis and control.
\end{IEEEbiography}

\begin{IEEEbiography}[{\includegraphics[width=1in,height=1.25in,clip,keepaspectratio]{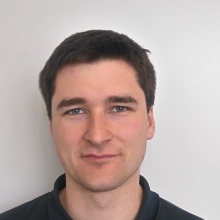}}]{Michael Hertneck} received the Master’s degree in Mechatronics from the University of Stuttgart, Stuttgart, Germany, in 2019. He has since been a research and teaching assistant at the Institute for Systems Theory and Automatic Control and a member of the Graduate School Simulation Technology at the University of Stuttgart. His research interests include Networked Control Systems with a focus on time- and event-triggered sampling strategies.
\end{IEEEbiography}

\begin{IEEEbiography}[{\includegraphics[width=1in,height=1.25in,clip,keepaspectratio]{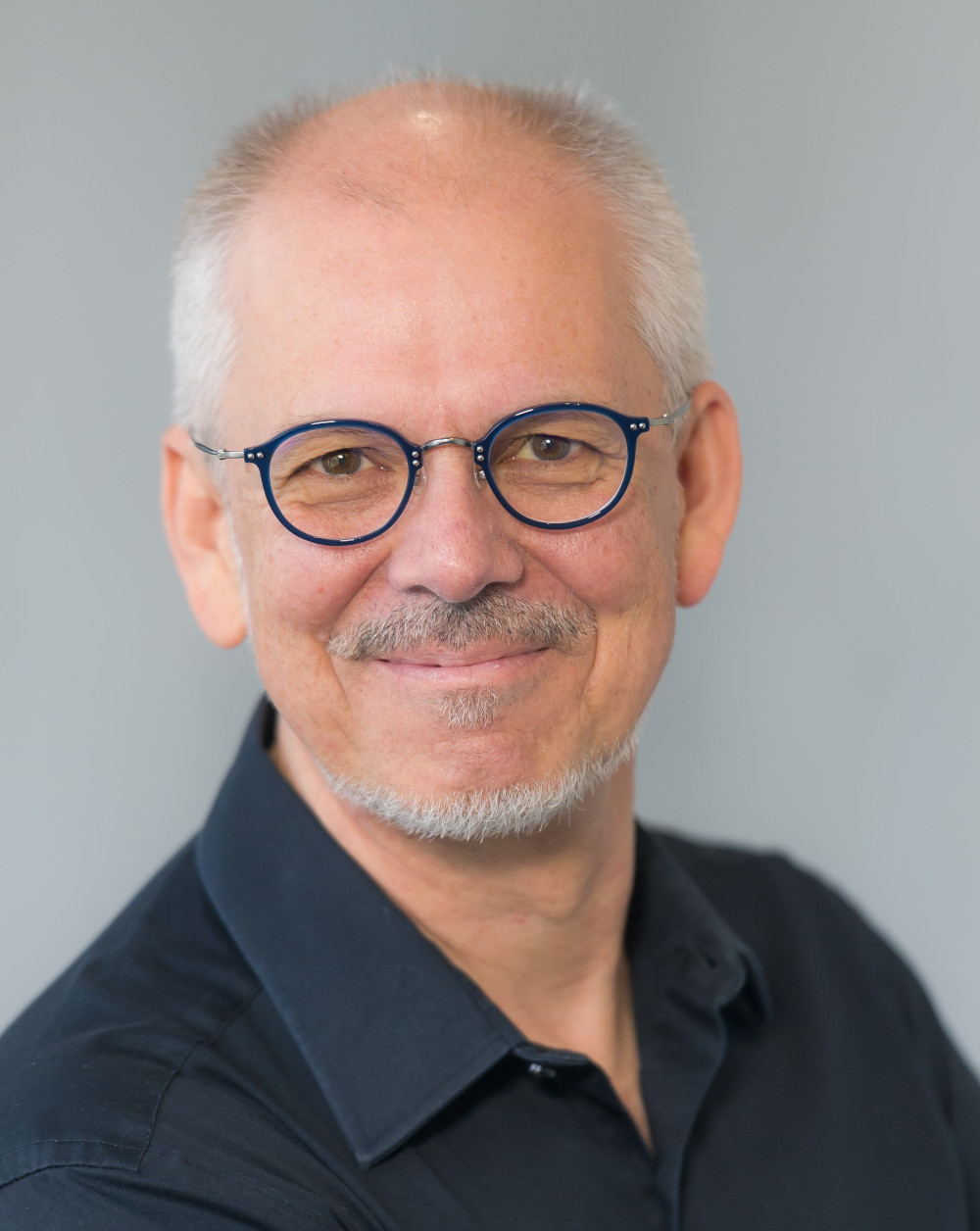}}]{Frank Allg\"ower} is professor of mechanical engineering at the University of Stuttgart, Germany, and Director of the Institute for Systems Theory and Automatic Control (IST) there.
	
He is active in serving the community in several roles: Among others he has been President of the International Federation of Automatic Control (IFAC) for the years 2017-2020, Vicepresident for Technical Activities of the IEEE Control Systems Society for 2013/14, and Editor of the journal Automatica from 2001 until 2015. From 2012 until 2020 he served in addition as Vice-president for the German Research Foundation (DFG), which is Germany’s most important research funding organization.
	
His research interests include predictive control, data-based control, networked control, cooperative control, and nonlinear control with application	to a wide range of fields including systems biology.
\end{IEEEbiography}

\end{document}